\documentclass[review]{elsarticle}

\usepackage{lineno,hyperref,mathtools}
\usepackage[letterpaper,top=2cm,bottom=2cm,left=3cm,right=3cm,marginparwidth=1.75cm]{geometry}

\journal{Journal of \LaTeX\ Templates}

\usepackage{amssymb}
\usepackage{bm}
\usepackage{amssymb}
\usepackage{amsthm}
\usepackage{multirow,booktabs}
\usepackage{cases}
\usepackage{threeparttable}

\newcommand{\F}{\mathbb{F}}
\newcommand{\Z}{\mathbb{Z}}
\newcommand{\C}{\mathcal{C}}
\newcommand{\lcm}{\mathrm{lcm}}
\newcommand{\ord}{\mathrm{ord}}
\newcommand{\MinRep}{\mathrm{MinRep}}

\newtheorem*{question}{{Question}}
\newtheorem{theorem}{Theorem}

\newtheorem{remark}{Remark}

\newtheorem{lemma}[theorem]{Lemma}
\newtheorem{corollary}[theorem]{Corollary}
\newtheorem{example}[theorem]{Example}
\begin{document}

	\begin{frontmatter}
		
		\title{Three classes of BCH codes and their duals}
	\tnotetext[mytitlenote]{The work was supported by the National Natural Science Foundation of China (12271137,12271335).}
		
		\author[mymainaddress]{Yanhui Zhang}
		\ead{zhangyanhui0327@163.com}
		
		\author[mymainaddress]{Li Liu\corref{mycorrespondingauthor}}
		\cortext[mycorrespondingauthor]{Corresponding author}
		\ead{liuli-1128@163.com}
		
		\address[mymainaddress]{School of Mathematics, Hefei University of Technology, Hefei 23009, Anhui, China}
		\author[myaddress]{Xianhong Xie}
		\ead{xianhxie@ahau.edu.cn}
        \address[myaddress]{School of information and computer science, Anhui Agricultural University, Hefei 230036, Anhui, China}
        
		\begin{abstract}
			BCH codes are an important class of cyclic codes, and have wide applicantions in communication and storage systems. However, it is difficult to determine the parameters of BCH codes and only a few cases are known. In this paper, we mainly study three classes of BCH codes with $n=q^{m}-1,\frac{q^{2s}-1}{q+1},\frac{q^{m}-1}{q-1}$. On the one hand, we accurately give the parameters of $\C_{(q,n,\delta,1)}$ and its dual codes. On the other hand, we give the sufficient and necessary conditions for $\C_{(q,n,\delta,2)}$ being dually-BCH codes.
		\end{abstract}
		
	\end{frontmatter}
	
	\section{Introduction}
	Let $p$ be a prime and $q>1$ be a $p$-power.
	An $[n,k,d]$ linear code $\C$ over $\F_q$ is a $k$-dimensional	subspace of $\F_q^{n}$ with minimum Hamming distance $d$.  $\C$ is said to be cyclic if
	$(c_{0}, c_{1}, ..., c_{n-1}) \in \C$ implies $(c_{n-1}, c_{0}, c_{1}, ..., c_{n-2})\in \C$. Identify any vector $(c_{0}, c_{1}, ..., c_{n-1})\in \F_q^{n}$ with
	\begin{center}
		$c_{0}+c_{1}x+c_{2}x^{2}+...+c_{n-1}x^{n-1}\in \F_q[x]/ \langle x^{n}-1\rangle$,
	\end{center}
	i.e., a code $\C$ of length $n$ over $\F_q$ corresponds to a subset of $\F_q[x]/\langle x^{n}-1\rangle$. Thus $\C$
	is a cyclic code if and only if the corresponding subset is an ideal of $\F_q[x]/\langle x^{n}-1\rangle$.
	Note that $\F_q[x]/\langle x^{n}-1\rangle$ is a principal ideal domain, this means that  there exists a monic polynomial
	$g(x)$ of the smallest degree such that $\C$=$\langle g(x) \rangle$ and $g(x)|(x^{n}-1)$. $g(x)$ is called the generator polynomial of $\C$, and $h(x) = (x^{n}-1)/g(x)$ is called the check polynomial of $\C$.
	
	For a code $\C$, its dual code, denoted by $\C^{\perp}$, is defined by
	\begin{center}
		$\C^{\perp}:=\lbrace \mathbf{b}\in \F_{q}^{n}:\ \mathbf{b}\cdot\mathbf{c}^{T}=0 \ \text{for all}\ c\in \C \rbrace$,
	\end{center}  where $^T$ denotes the transport and $\cdot$ denotes the standard inner product.
	
	Let $m = \ord_{n}(q)$, $\F_{q^m}^*=\langle\alpha\rangle$  and $\beta= \alpha^{\frac{q^{m}-1}{n}}$. Then $\beta$ is a primitive $n$-th root of unity.
	Let $m_{i}(x)$ denote the minimal polynomial of $\beta ^{i}$ over $\F_q$, $0\leq i\leq n-1$. For positive integers $b$ and $\delta$, define	
	\begin{center}
		$g_{(q,n,\delta,b)}(x):=\lcm \left(m_{b}(x),m_{b+1}(x),...,m_{b+\delta -2}(x)\right)$,
	\end{center}
	where $2\leq \delta \leq n$, $\lcm$ denotes the least common multiple of $m_i(x)$, $b\leq i\leq b+\delta-2$. Then $\C_{(q,n,\delta,b)}$ is called a BCH code of length $n$ with designed distance $\delta$. For $b=1$, $\C_{(q,n,\delta)}:=\C_{(q,n,\delta,1)}$ is  narrow-sense BCH code. Suppose $d(\C_{(q,n,\delta,b)})$ denotes the minimum distance of $\C_{(q,n,\delta,b)}$, then $d(\C_{(q,n,\delta,b)})\geq \delta$.
	
	The importance of the BCH codes in coding theory and communication is apparent, as can be seen in  \cite{RefJ5} and \cite{RefJ13}. However, there are many interesting problems about BCH codes or coset leaders.
	In \cite{RefJ12}, Gong et al. proposed the following question:
	\begin{question}
		For a BCH code over $\F_q$, when is its dual code BCH with respect to the same primitive root of unity?
	\end{question}
	If the dual of BCH code $\C$ is still a BCH code with respect to the same primitive root of unity, then $\C$ is called a dually-BCH code. In \cite{RefJ12}, the authors gave a necessary and sufficient condition for $\C_{(q,q^{m}-1,\delta)}$ being a dually-BCH code. In \cite{RefJ19}, Wang et al. presented a necessary and sufficient condition for $\C_{(q,\frac{q^{m}-1}{q-1},\delta)}$ and $\C_{(q,\frac{q^{2s}-1}{q+1},\delta)}$ being a dually-BCH codes. But for other lengths or $b>1$ and $n=q^m-1,\frac{q^{2s}-1}{q+1},\frac{q^{m}-1}{q-1}$, the condition for $\C_{(q,n,\delta,b)}$ being a dually-BCH code is still unknown.

	It is well known that there is a close relationship between cyclotomic coset leaders and BCH codes. For this reason, many authors determined the cyclotomic cosets for the study of BCH codes, and obtained some pretty results.
	\begin{itemize}
		\item For $n=q^m-1$, Ding et al.\cite{RefJ8} presented the first three largest $q$-cyclotomic coset leaders modulo $n=q^{m}-1$. Using those coset leaders, the authors constructed several BCH codes and gave their dimensions and distances (see \cite{RefJ6}-\cite{RefJ10},\cite{RefJ20}).
		\item For $m=2s$ and $n=\frac{q^{2s}-1}{q+1}$, Wu et al. \cite{RefJ18} gave the largest coset leader $\delta_{1}$ modulo $n$ for $q=2,3$, and determined the dimension of $\C_{(q,n,\delta)}$ for $2\nmid s$ and $2\leq \delta \leq q^{s}+1$ or $2\mid s$ and $2\leq \delta\leq \lceil \frac{q}{2}\rceil q^{s-1}+1$. Recently, Wang et al. \cite{RefJ19} proved that $\delta_1$ is still the largest $q$-cyclotomic coset leaders modulo $n$ for $q\geq4$.
		\item For $n=\frac{q^{m}-1}{q-1}$, Ding et al. \cite{RefJ17} gave the first two largest coset leaders $\delta_1$ and $\delta_2$ for $q=3$, and determined the parameters of $\C_{(q,n,\delta)}$
		for $\delta=\delta_1$ and $\delta_2$. Zhu et al. \cite{RefJ21} gave the largest coset
		leaders $\delta_1$ for $q>3$ and $m-1\equiv 0,1,q-2\bmod{q-1}$, and determined the parameters
		of $\C_{(q,n,\delta_1)}$. Wang et al. \cite{RefJ19} presented the largest $q$-cyclotomic coset leader for any $q$, and gave the dimension of $\C_{(q,n,\delta_1)}$.
	\end{itemize}
	
	We shall work on the $q$-cyclotomic coset leaders modulo $n=q^m-1,\frac{q^{2s}-1}{q+1},\frac{q^{m}-1}{q-1}$, and the corresponding BCH codes. Our main contributions are:
	\begin{itemize}
		\item [(A)]When $n=q^m-1$. Set $\widetilde{\C}_{(q,q^{m}-1,\delta,b)}:=\langle (x-1)g_{(q,n,\delta,b)}(x)\rangle$. We prove that the $i$-th largest $q$-cyclotomic coset leader is $\delta_{i}=(q-1)q^{m-1}-1-q^{\lfloor \frac{m-1}{2} \rfloor +i-2}$ for $m-(\left\lfloor \frac{m-1}{2} \right\rfloor +\left\lfloor \frac{m-3}{3} \right\rfloor)\geq i\geq 3$, and give a sufficient and necessary condition for $\widetilde{\C}_{(q,q^{m}-1,\delta)}^{\perp}$ being a narrow-sense primitive BCH code. Furthermore, when $n\in\{q^m-1,\frac{q^{2s}-1}{q+1},\frac{q^{m}-1}{q-1}\}$, we give a sufficient and necessary condition for $\C_{(q,n,\delta,2)}$ being dually-BCH codes.
		\item[(B)] When $n=\frac{q^{2s}-1}{q+1}$, we obtain the second largest coset leader $\delta_{2}$ modulo $n$. Moreover,\begin{itemize}
			\item [(B.1)]For $ \delta_{2}\leq \delta\leq \delta_{1}$, we completely give the dimension and minimum distance of $\C_{(q,n,\delta)}$ and $\C_{(q,n,\delta)}^{\perp}$.
			\item [(B.2)] For $2\mid s$ and $\lceil \frac{q}{2}\rceil q^{s-1}+1\leq \delta \leq \frac{q^{s+1}+1}{q+1}$, we give the dimension of $\C_{(q,n,\delta)}$.
			\item [(B.3)] For $\delta=a\frac{q^{s}-1}{q-1},a\frac{q^{s}+1}{q+1}$ if $2\nmid s$ and $\delta=a\frac{q^{s}-1}{q^{2}-1}$ if $2\mid s$, $1\leq a\leq q-1$, we give the dimension and minimum distance of $\C_{(q,n,\delta)}$.
		\end{itemize}
		\item [(C)]When $n=\frac{q^{m}-1}{q-1}$. We obtain the second largest coset leader $\delta_{2}$ for some special cases, and present the dimension of $\C_{(q,n,\delta)}$ for $ \delta_{2}\leq \delta\leq \delta_{1}$.
		
	\end{itemize}

	\section{Preliminaries}
	\subsection{Basic Notations}
	For any positive integer $0\leq s\leq q^m-2$, its $q$-adic expansion is  $s=\sum_{j=0}^{m-1}s_{j}q^{j}$, write $s=(s_{m-1},s_{m-2},...,s_{1},s_{0})$. For integer $0\leq i\leq m-1$, we denote $sq^i\pmod{n}$ by $[sq^{i}]_{n}$. Then if $n=q^{m}-1$, we have
	\begin{center}
		$[sq^{i}]_{q^m-1}:=(s_{m-1-i},...,s_{0},s_{m-1},...,s_{m-i})$.
	\end{center}
	
For any $1\leq i\leq n-1$, $\delta_{i,n}$ denotes the $i$-th largest $q$-cyclotomic coset leader modulo $n$.	

	Let $T=\lbrace 0\leq i\leq n-1:g_{(q,n,\delta,b)}(\beta^{i})=0\rbrace$ and $T^{-1}=\lbrace n-i:i\in T\rbrace$. Then $T$ and $T^{\perp}=\Z_{n}$$\backslash T^{-1}$ are called the defining sets of $\C_{(q,n,\delta,b)}$ and $\C^{\perp}_{(q,n,\delta,b)}$ with respect to $\beta $, respectively.
	
	\subsection{Cyclotomic Cosets and Coset Leaders}
	For any $t$ with $0\leq t\leq n-1$, the  set
	\[\{tq^i\pmod{n}: \ 0\leq i< m\}\]
	is called the $q$-cyclotomic coset modulo $n$ of representative $t$ and is denoted by $C_t$. The number of elements in $C_t$ is denoted by $\mid C_t\mid$. Set $CL(t):=\min\{i:\ i\in C_t\}$ and $\MinRep_{n}:=\{CL(t): \ 0\leq t\leq n-1\}$. Then $CL(t)$ is called the coset leader of $C_{t}$.
	
	It is well known that the coset leaders are very important to evaluate the dimension and minimum distance of BCH codes. The following four lemmas on coset leaders modulo $q^m-1$, $\frac{q^{2s}-1}{q+1}$ and $\frac{q^m-1}{q+1}$ will be useful for demonstrating our results.
	\begin{lemma}(\cite{RefJ2})\label{lem1}
	Let $n$ be a positive integer such that $\gcd(n,q)=1$ and $q^{\lfloor \frac{m}{2}\rfloor}<n\leq q^{m}-1$. Then $s$ is a coset leader and $\mid C_{s}\mid=m$ for all $1\leq s\leq \frac{nq^{\lceil \frac{m}{2}\rceil}}{q^{m}-1}$, $s\not\equiv0\pmod{q}$.
   \end{lemma}
	\begin{lemma}(\cite{RefJ14,RefJ16})\label{lem2} Let $n=q^m-1$. Then
		\begin{itemize}
			\item [(a)] The first three largest $q$-cyclotomic coset leaders modulo $n$ are:
			\begin{center}
				$\delta_{1,n}=(q-1)q^{m-1}-1$, $\delta_{2,n}=(q-1)q^{m-1}-1-q^{\lfloor \frac{m-1}{2}\rfloor}$, $\delta_{3,n}=(q-1)q^{m-1}-1-q^{\lfloor \frac{m+1}{2}\rfloor}$.
			\end{center}
		\item [(b)]If the Bose distance of $\C_{(n,q,\delta)}$ is $d_{i}=q^{m}-q^{m-1}-q^{i}-1$, where $\frac{m-2}{2}\leq i\leq m-\lfloor \frac{m}{3}\rfloor-1$. Then the minimum distance of $\C_{(n,q,\delta)}$ is $d_{i}$.
			\item [(c)]Let $m=2s$ and $a$ be an integer satisfying $q^{s}+1\leq a\leq q^{s+1}$ and $a\not\equiv0\pmod{q}$.
			\begin{itemize}
				\item [(c.1)] Set $a=c(q^{s}+1)$, $1\leq c\leq q-1$, then $a$ is a coset leader and $\mid C_{a}\mid =\frac{m}{2}$.
				\item [(c.2)] Set $a=a_{s}q^{s}+a_{0}$, $1\leq a_{0}<a_{s}\leq q-1$, then $a$ is not a coset leader.
				\item [(c.3)] Except for (b.1) and (b.2), the remaining of $a$ are coset leaders and  $\mid C_{a}\mid =m$.
			\end{itemize}
			
		\end{itemize}
		
	\end{lemma}

	\begin{lemma}(\cite{RefJ19})\label{lem3}
		Let $n=\frac{q^{m}-1}{q+1}$ and $m=2s$, the first largest $q$-cyclotomic coset leader $\delta_{1,n}$ modulo $n$ is:
		\begin{itemize}
			\item 	If $2\nmid s$, then $\delta_{1,n} =\frac{(q-1)q^{m-1}-q^{\frac{m-2}{2}}-1}{q+1}$ and $\mid C_{\delta_{1,n}}\mid =\frac{m}{2}$.
			\item If $2\mid s$, then $\delta_{1,n}=
			\frac{(q-1)q^{m-1}-q^{\frac{m}{2}}-1}{q+1}$ and $\mid C_{\delta_{1,n}}\mid =m$.
		\end{itemize}
	Furthermore, set $(q+1)\mid h$. Then $h$
	is a coset leader modulo $q^{m}-1$ if and only if $\frac{h}{q+1}$ is a coset leader modulo $\frac{q^m-1}{q+1}$.
	\end{lemma}
	\begin{lemma}(\cite{RefJ19},\cite{RefJ21})\label{lem4}
		Let $n=\frac{q^m-1}{q-1}$. Then
		\begin{itemize}
			\item [(a)] Let	$q>3$. For any integer $1\leq i\leq n-1$, take $i=(i_{m-1},i_{m-2},\ldots,i_0)$. If $i$ is a $q$-cyclotomic coset leader modulo $\frac{q^{m}-1}{q-1}$, then $i_{m-1}=0$.
			
			Furthermore, suppose $m-1=a(q-1)+b$, where $a\geq 1$ and $0\leq b\leq q-2$. Let $\epsilon=a+1$ when $b=q-2$ and $\epsilon=a$ when $0\leq b\leq q-3$. If $i_{l}=q-1$ for all $m-1-\epsilon\leq l \leq m-2$, then $1\leq i_{l-1}\leq i_{l}$ for all $1\leq l \leq m-2$.
			\item [(b)] Let $q\geq 3$ and $m\geq4$.
			\begin{itemize}
				\item [(b.1)] Let $q-1=mt_{1}+t_{2}$ and $\sum_{t=1}^{q-1}q^{\left\lceil \frac{mt}{q-1}-1\right\rceil}=\sum_{i=0}^{m-1}a_{i}q^{i}$, where $t_{1}\geq 0$ and $m>t_{2}\geq 0$.
				\begin{itemize}
					\item If $t_{2}=0$, then $a_{i}=\frac{q-1}{m}$ for all $i\in [0,m-1]$.
					\item If $t_{2}\neq0$, then $a_{i}=\lceil\frac{q-1}{m} \rceil$ when $i\in \Upsilon$, where $\Upsilon=\lbrace\lceil\frac{m\gamma}{t_{2}}-1\rceil,\gamma=1,2,\ldots,t_{2}\rbrace$. Otherwise, $a_{i}=\lfloor\frac{q-1}{m} \rfloor$.
				\end{itemize}
				\item [(b.2)] The first largest $q$-cyclotomic coset leader modulo $n$ is $\delta_{1,n}=\frac{q^{m}-1-\sum_{t=1}^{q-1}q^{\lceil \frac{mt}{q-1}-1\rceil}}{q-1}$ and
				\\ $\mid C_{\delta_{1,n}}\mid =\frac{m}{gcd(m,q-1)}$.
			\end{itemize}
		\end{itemize}
	\end{lemma}

	\subsection{Known Results on the Dimension and Minimum Distance of BCH Codes}
	Given designed distance $\delta$, it is difficult  to determine the dimension  and minimum distance. But for some special BCH codes, $\dim(\C_{(q,n,\delta)})$ and $d(\C_{(q,n,\delta)})$ can be given. We list them as follows.
	\begin{lemma}(\cite{RefJ17})\label{lem5}
		Let $n$ be a positive integer such that  $q-1\mid n$ and $\gcd(n,q)=1$, let $\delta_{b}$ be a divisor of $ \frac{n}{q-1}$. Then for $\delta=k\delta_{b}$, $1\leq k\leq q-1$, the minimum distance of $\C_{(n,q,\delta)}$ is $\delta$.
	\end{lemma}
	\begin{lemma}(\cite{RefJ18})\label{lem6}
		Let $n=\frac{q^{m}-1}{q+1}$ and $m=2s$. Then
		\begin{itemize}
			\item [(a)] Suppose $2\leq \delta \leq q^{s}+1$. For $q,s\geq3$ and $s$ is odd or $q=2$ and $s\geq5$, the dimension $k$ of $\C_{(q,n,\delta,1)}$ is given as follows:
			\begin{itemize}
				\item [(a.1)] If $2\leq \delta \leq \frac{q^{s}+1}{q+1}$, then $$k=n-2s(\delta-1)+2s\left\lfloor \frac{\delta-1}{q} \right\rfloor.$$
				\item [(a.2)] If $\frac{q^{s}+1}{q+1}+1\leq \delta \leq (q-1)\frac{q^{s}+1}{q+1}+1$, then $$k=n-2s(\delta-1)+2s\left\lfloor \frac{\delta-1}{q} \right\rfloor+s\left\lfloor\frac{(\delta-1)(q+1)}{q^{s}+1}\right\rfloor.$$
				\item [(a.3)] If $(q-1)\frac{q^{s}+1}{q+1}+2\leq \delta \leq \frac{q^{s+1}-1}{q+1}+2$, then $$k=n-2s(\delta-1)+2s\left\lfloor \frac{\delta-1}{q} \right\rfloor+s(q-1).$$
				\item [(a.4)] If $(q-1)q^{s-1}+\frac{q^{s}+1}{q+1}+1\leq \delta \leq q^{s}+1$, then $$k=n-2s(\delta-1)+2s\left\lfloor \frac{\delta-1}{q} \right\rfloor+3s(q-1).$$
			\end{itemize}
			\item [(b)] Suppose   $2\leq \delta \leq \lceil \frac{q}{2}\rceil q^{s-1}+1$.   For $s\geq 4$ and is even, the dimension $k$ of $\C_{(q,n,\delta)}$ is given as follows:
			\begin{itemize}
				\item [(b.1)] If $2\mid q $, then $$k=n- 2s(\delta-1)+ 2s\left\lfloor \frac{\delta-1}{q} \right\rfloor.$$
				\item [(b.2)] If $2\nmid q$, then
				\begin{center}
					$k$ =$\begin{cases}
						n- 2s(\delta-1)+ 2s\left\lfloor \frac{\delta-1}{q} \right\rfloor,& \text{if $2\leq \delta \leq \frac{q^{s}+1}{2}$}; \\
						n- 2s(\delta-1)+ 2s\left\lfloor \frac{\delta-1}{q} \right\rfloor+s,& \text{if $\frac{q^{s}+1}{2}+1\leq \delta \leq \frac{q+1}{2}q^{s-1}+1$}.
					\end{cases}$
				\end{center}
			\end{itemize}
		\end{itemize}
	\end{lemma}
	
	\subsection{Known Results on Dually-BCH Codes}
For special designed distance of narrow-sense  BCH codes,
		the following two lemmas give the sufficient and necessary conditions on dually-BCH codes.
	\begin{lemma}(\cite{RefJ18})\label{lem7}
		Let $n=q^{m}-1$.
		\begin{itemize}
			\item If $q=2$ and $m\geq 6$, then $\C_{(2,2^{m}-1,\delta)}$ is a dually-BCH code if and only if
			\begin{center}
				$\delta=2,3,$ or $2^{m-1}-2^{\lfloor \frac{m-1}{2} \rfloor}\leq \delta \leq n$.
			\end{center}
			\item If $q\geq 3$ and $m\geq 2$, then $\C_{(q,q^{m}-1,\delta)}$ is a dually-BCH code if and only if
			\begin{center}
				$\delta=2$ or $(q-1)q^{m-1}-q^{\lfloor \frac{m-1}{2} \rfloor}\leq \delta \leq n$.
			\end{center}
		\end{itemize}
	\end{lemma}
	
	\begin{lemma}(\cite{RefJ19})\label{lem8}
		Let $n=\frac{q^{m}-1}{q+1}$ and $\delta_{1,n}$ is given in Lemma \ref{lem3}.
			\begin{itemize}
				\item [(1)]If $q=2$ and $m\geq 4$ is even, then $\C_{(q,n,\delta)}$ is a dually-BCH code if and only if $$\delta_{1,n}+1\leq \delta\leq n .$$
				\item [(2)] If $q>2$ and $m=4$, then $\C_{(q,n,\delta)}$ is a dually-BCH code if and only if $$\delta=2,\,\delta_{1,n}\leq \delta\leq n .$$
				\item [(3)] If $q>2$ and $m\neq4$ is even, then $\C_{(q,n,\delta)}$ is a dually-BCH code if and only if $$\delta_{1,n}+1\leq \delta\leq n .$$
			\end{itemize}
		In addition, let $q\geq 3$, $m\geq 4$, $n=\frac{q^{m}-1}{q-1}$ and $\delta_{1,n}$ is given in Lemma \ref{lem4}. Then $\C_{(q,n,\delta)}$ is a dually-BCH code if and only if $$\delta_{1,n}+1\leq \delta\leq n.$$
	
	\end{lemma}

\section{The case of $n=q^{m}-1$}
In this section, we always assume $n=q^m-1$.
\subsection{The computation of $\delta_{i,n}$}

\begin{lemma}\label{lem9}
Let $q$ be a prime power and $3\leq i\leq m-(\left\lfloor \frac{m-1}{2} \right\rfloor +\left\lfloor \frac{m-3}{3} \right\rfloor)$,
	\[\delta_{i,n}=(q-1)q^{m-1}-1-q^{\lfloor \frac{m-1}{2} \rfloor +i-2},\]
and $\mid C_{\delta_{i,n}}\mid=m$.
\begin{proof}
Obviously, the case of $i=3$ agrees with
	 Lemma \ref{lem2}. By induction, suppose
	 $\delta_{t,n}=(q-1)q^{m-1}-1-q^{\lfloor \frac{m-1}{2} \rfloor +t-2}$, then we need to prove $\delta_{t+1,n}=(q-1)q^{m-1}-1-q^{\lfloor \frac{m-1}{2} \rfloor +t-1}$ for $m-(\left\lfloor \frac{m-1}{2} \right\rfloor +\left\lfloor \frac{m-3}{3} \right\rfloor)\geq t+1>t\geq 2$. In order to prove this, we divide it into two steps.

Step 1. We claim that  $\delta_{i,n}$ is a coset leader for $3\leq i\leq m-(\left\lfloor \frac{m-1}{2} \right\rfloor +\left\lfloor \frac{m-3}{3} \right\rfloor)$. Clearly,
$$\delta_{i,n}=(q-2,\underbrace{q-1,\ldots,q-1}_{m-\left\lfloor\frac{m-1}{2}\right\rfloor-i},q-2,\underbrace{q-1,\ldots,q-1}_{\left\lfloor \frac{m-1}{2}\right\rfloor+i-2}) .$$
Since $i\geq 3$, then $m-\left\lfloor \frac{m-1}{2}\right\rfloor-i< \left\lfloor \frac{m-1}{2}\right\rfloor+i-2$. Note that
\begin{center}
	$\begin{cases}
		[\delta_{i}q]_{n}=(\underbrace{q-1,\ldots,q-1}_{m-\left\lfloor\frac{m-1}{2}\right\rfloor-i},q-2,\underbrace{q-1,\ldots,q-1}_{\left\lfloor \frac{m-1}{2}\right\rfloor+i-2},q-2),\\
		[\delta_{i}q^{2}]_{n}=(\underbrace{q-1,\ldots,q-1}_{m-\left\lfloor\frac{m-1}{2}\right\rfloor-i-1},q-2,\underbrace{q-1,\ldots,q-1}_{\left\lfloor \frac{m-1}{2}\right\rfloor+i-2},q-2,q-1),\\
		\cdots\\
		[\delta_{i}q^{m-\left\lfloor\frac{m-1}{2}\right\rfloor-i}]_{n}=(q-2,\underbrace{q-1,\ldots,q-1}_{\left\lfloor\frac{m-1}{2}\right\rfloor+i-2},q-2,\underbrace{q-1,\ldots,q-1}_{m-\left\lfloor\frac{m-1}{2}\right\rfloor-i}),\\
		\cdots\\
		[\delta_{i}q^{m-1}]_{n}=(q-1,q-2,\underbrace{q-1,\ldots,q-1}_{m-\left\lfloor\frac{m-1}{2}\right\rfloor-i},q-2,\underbrace{q-1,\ldots,q-1}_{\left\lfloor \frac{m-1}{2}\right\rfloor+i-3}).
	\end{cases}$
\end{center}
Then  $[\delta_{i,n}q^{j}]_{n}> \delta_{i,n}$ for any $1\leq j\leq m-1$, this implies that $\delta_{i,n}$ is a coset leader and $\mid C_{\delta_{i,n}}\mid=m$ for $3\leq i\leq m-(\left\lfloor \frac{m-1}{2} \right\rfloor +\left\lfloor \frac{m-3}{3} \right\rfloor)$.

Step 2: For $m-\left\lfloor \frac{m-1}{2} \right\rfloor-t> \left\lfloor \frac{m-3}{3} \right\rfloor$, we claim that $J_{i}:=\delta_{t,n}-i$ is not a coset leader for any $1\leq i\leq (q-1)q^{\lfloor \frac{m-1}{2} \rfloor +t-2}-1$.

Take $h=\lfloor \frac{m-1}{2} \rfloor +t-2$. Note that $\delta_{t,n}-\delta_{t+1,n}=q^{h+1}-q^{h}=(q-1)q^{h}$ and
\begin{center}
	$(q-1)q^{h}-1=(q-2)q^{h}+(q-1)q^{h-1}+(q-1)q^{h-2}+\cdots+(q-1)q+q-1$.
\end{center}
For any $1\leq i\leq (q-1)q^{h}-1$, let
\begin{center}
	$i:=i_{h}q^{h}+i_{h-1}q^{h-1}+\cdots+i_{1}q+i_{0}$,
\end{center}
then $0\leq i_{j}\leq q-1$ for any $0\leq j\leq h-1$ and $0\leq i_{h}\leq q-2$, and there is at least an $i_j\neq0$ for $j\in\{0,1,\ldots,h\}$.
Thus,
	\begin{align}	J_{i}=&(q-2)q^{m-1}+(q-1)q^{m-2}+\cdots+(q-1)q^{h+1}+(q-2-i_{h})q^{h}
		+(q-1-i_{h-1})q^{h-1}\notag \\
		&+\cdots+(q-1-i_{1})q+(q-1-i_{0}).\label{eq:1}
	\end{align}
	For $q=2$, we have $i_{h}=0$ and
\begin{center}
	$J_{i}=2^{m-2}+2^{m-3}+\cdots+2^{h+1}+(1-i_{h-1})2^{h-1}+\cdots+(1-i_{1})2+1-i_{0}$.
\end{center}

If $i_{0}=1$, then $\frac{J_{i}}{2}$ and $J_{i}$ are in the same cyclotomic coset and $\frac{J_{i}}{2}<J_{i}$. Hence, $J_{i}$ cannot be a coset leader.

If $i_{0}=0$, suppose $l$ is the largest integer such that $i_{l}=1$, $1\leq l\leq t-1$. Then
\begin{center}
	$J_{i}=(0,\underbrace{1,1,\ldots,1}_{m-2-h},0,\underbrace{1,1,\ldots,1}_{h-l-1},0,\underbrace{1-i_{l-1},\ldots,1-i_{1},1}_{l})$.
\end{center}
Note that $m-2-h=m-\left\lfloor\frac{m-1}{2}\right\rfloor-t $ and $h-l-1=\left\lfloor \frac{m-1}{2}\right\rfloor+t-3-l$.
Since $m-(\left\lfloor \frac{m-1}{2} \right\rfloor +\left\lfloor \frac{m-3}{3} \right\rfloor)\geq t+1$, then $m-\left\lfloor \frac{m-1}{2} \right\rfloor-t\geq \left\lfloor \frac{m-3}{3} \right\rfloor+1> \left\lfloor \frac{m-3}{3} \right\rfloor$.

Note that
\begin{equation}\label{eq:2}
	h-l+1+l=m-3-(m-\left\lfloor \frac{m-1}{2} \right\rfloor-t)<m-3-\left\lfloor \frac{m-3}{3} \right\rfloor\leq 2\left\lfloor\frac{m-3}{3}\right\rfloor.\end{equation}
Then $h-l+1\leq \left\lfloor\frac{m-3}{3}\right\rfloor$ if $h-l+1\leq l$ or $l\leq \left\lfloor\frac{m-3}{3}\right\rfloor$ if $h-l+1\geq l$, then
 $m-2-h> h-l-1$ or $m-2-h>l$.

If $m-2-h> h-l-1$, then
\begin{equation}\label{eq:3}
	[J_{i}q^{m-1-h}]_{q^{m}-1}=
(0,\underbrace{1,1,\ldots,1}_{h-l-1},0,\underbrace{1-i_{l-1},
\ldots,1-i_{1},1}_{l},0,\underbrace{1,1,\ldots,1}_{m-2-h})<J_i.
\end{equation}

If $m-2-h>l$, then
\begin{equation}\label{eq:4}
	[J_{i}q^{m-1-l}]_{q^{m}-1}=(0,\underbrace{1-i_{l-1},
\ldots,1-i_{1},1}_{l},0,\underbrace{1,1,\ldots,1}_{m-2-h},0,\underbrace{1,1\ldots,1}_{h-l-1})<J_i.
\end{equation}
By Eqs.~\eqref{eq:3} and ~\eqref{eq:4}, we know that $J_{i}$ cannot be a coset leader.
\\For $q>2$, by Eq.~\eqref{eq:1},
\begin{itemize}
  \item If $i_{h}\geq 1$, we have $q-2>q-2-i_{h}$, then $J_{i}q^{m-1-h}\bmod{q^{m}-1}<J_{i}$.
  \item If there exists an integer $j$ such that $i_{j}\geq 2$, $0\leq j\leq t-1$, we have $q-2>q-1-i_{j}$, then $J_{i}q^{m-1-j}\bmod {q^{m}-1}<J_{i}$.
\end{itemize} Thus, $J_{i}$ cannot be a coset leader. Next we consider $i_{h}=0$ and $i_{j}\in \lbrace 0,1\rbrace$ for any $0\leq j\leq t-1$.

Note that $i\geq1$, suppose $l$ is the largest index such that $i_{l}=1$, $1\leq l\leq t-1$. Then
\begin{equation*}
	J_{i}=(q-2,\underbrace{q-1,\ldots,q-1}_{m-2-h},q-2,\underbrace{q-1,\ldots,q-1}_{h-l-1},	q-2,\underbrace{q-1-i_{l-1},\ldots,q-1-i_{1},q-1-i_{0}}_{l}).
\end{equation*}
By Eqs.~\eqref{eq:2},~\eqref{eq:3} and ~\eqref{eq:4},
$J_{i}$ cannot be a coset leader.
Thus we complete the proof.
\end{proof}	
\end{lemma}
In particular, for $i=4,5$, we have
\begin{corollary}\label{cor10}
	\begin{itemize}
		\item [(1)]Let $m\geq 10$. Then  $\delta_{4,n}=(q-1)q^{m-1}-1-q^{\lfloor \frac{m-1}{2} \rfloor +2}$ and $\mid C_{\delta_{4,n}}\mid=m$.
		\item [(2)]Let $m\geq 14$. Then $\delta_{5,n}=(q-1)q^{m-1}-1-q^{\lfloor \frac{m-1}{2} \rfloor +3}$ and $\mid C_{\delta_{5,n}}\mid=m$.
	\end{itemize}
\end{corollary}
\begin{remark}
	The proof of Lemma \ref{lem9} has been given in \cite{RefJ14}. For completeness, we provide a proof which is different from \cite{RefJ14}, and we give $\mid \delta_{i,n}\mid $.

\end{remark}
\subsection{BCH Codes and Dually-BCH Codes}
The following theorem provides the information on the parameters of the BCH code $\C_{(q,n,\delta_{i,n})}$.
\begin{theorem}\label{th11}
   Let $q$ be a prime power and $3\leq i\leq m-(\left\lfloor \frac{m-1}{2} \right\rfloor +\left\lfloor \frac{m-3}{3} \right\rfloor)$, then $\C_{(q,n,\delta_{i,n})}$ has parameters $[q^{m}-1,k,\delta_{i,n}]$, where
\[k=\begin{cases}
  	im,& \text{if}\ 2\nmid m; \\
  	(i-\frac{1}{2})m,& \text{if}\ 2\mid m.
  \end{cases}\]
\begin{proof}
	Form Lemmas \ref{lem2} and \ref{lem9}, we can obtain the results directly.
\end{proof}
\end{theorem}

For $b=1$, the conditions of $\C_{(q,q^m-1,\delta)}$ being dually-BCH codes have been given by Lemma~\ref{lem7}. For $b=2$, we have

\begin{theorem}\label{th12}
Let $q\geq 3$ and $m\geq 2$. Then ${\C}_{(q,n,\delta,2)}$ is a dually-BCH code if and only if
\begin{center}
	 $(q-1)q^{m-1}-q^{\lfloor \frac{m-1}{2} \rfloor}-1\leq \delta \leq q^m-2$.
\end{center}
\end{theorem}

\begin{proof}
	By the definition, the defining set of ${\C}_{(q,q^{m}-1,\delta,2)}$ with respect to $\alpha $ is $T=C_{2}\cup C_{3}\cup\cdots\cup C_{\delta}$, $2\leq \delta\leq n-1$. Note that $0\notin T$, i.e., $0\in T^{\perp}$, this means $C_{0}\subset T^{\perp}$. Therefore, if $\C_{(q,n,\delta,2)}$ is a dually-BCH code, then there must exist an integer $r\geq 1$ such that $T^{\perp}=C_{0}\cup C_{1}\cup \cdots\cup C_{r-1}$.
	
If $q\leq \delta\leq n-1$, then $C_{1}=C_{q}\subset T$, i.e., $T=C_{1}\cup C_{2}\cup\cdots\cup C_{\delta}$. Hence, by Lemma \ref{lem7}, we know that ${\C}_{(q,q^{m}-1,\delta,2)}$ is a dually-BCH code if and only if $(q-1)q^{m-1}-q^{\lfloor \frac{m-1}{2} \rfloor}-1\leq \delta \leq n-1$.
	
If $2\leq \delta\leq q-1$. Since $\left[1,q-1\right]\subset\left[ 1,q^{\lceil\frac{m}{2} \rceil}\right]$, then $i\in \MinRep_{n}$ and $\mid C_{i}\mid=m$ for any $1\leq i\leq q-1$ by Lemma \ref{lem1}. We thus have $C_{2}\subset T$ and $C_{1}\not\subset T$, i.e., $m\leq \dim(\C_{(q,n,\delta,2)})\leq n-m<n$ and  $C_{CL(n-1)}\subset T^{\perp}$. Clearly, $CL(n-1)=(q-1)q^{m-1}-1=\delta_{1,n}$. Therefore, if $\C_{(q,n,\delta,2)}$ is a dually-BCH code, then $T^{\perp}=C_{0}\cup C_{1}\cup \cdots \cup C_{\delta_{1,n}}$. However, $\dim(\C_{(q,n,\delta,2)})\leq n-m$, which contradicts the fact that $\dim(\C_{(q,n,\delta,2)})+\dim(\C^{\perp}_{(q,n,\delta,2)})=n$. Thus we complete the proof.
\end{proof}

\begin{theorem}\label{th13}
Let $q=2$ and $m\geq 6$. Then ${\C}_{(q,n,\delta,2)}$ is a dually-BCH code if and only if
\begin{center}
	$\delta=2$ or $2^{m-1}-2^{\lfloor \frac{m-1}{2} \rfloor}-1\leq \delta \leq n-1$.
\end{center}

\end{theorem}
\begin{proof}
The proof is very similar to that of Theorem \ref{th12}, hence we omit it.	
\end{proof}
For $b=1$, we then study the condition of $\widetilde{\C}_{(2,n,\delta)}^{\perp}$ being a narrow-sense BCH code.
\begin{theorem}\label{th14}
Let $q=2$ and $m\geq 6$. Then ${\widetilde{\C}_{(2,n,\delta)}}^{\perp}$ is a narrow-sense BCH code if and only if
\begin{center}
	$\delta=2,3,$ or $2^{m-1}-2^{\lfloor \frac{m-1}{2} \rfloor}\leq \delta \leq 2^{m-1}-1$.
\end{center}
\end{theorem}

\begin{proof}
By definition, the defining set of $\widetilde{\C}_{(2,n,\delta)}$ with respect to $\alpha$ is $T=C_{0}\cup C_{1}\cup C_{2}\cup \cdots \cup C_{\delta-1}=T^{'}\cup C_{0}$, where $T'= C_{1}\cup C_{2}\cup \cdots \cup C_{\delta-1}$. Hence, \[T^{\perp}=\Z_{n}\backslash T^{-1}=(\Z_{n}\backslash T)^{-1}=((\Z_{n}\backslash (T^{'}))^{-1} )\backslash C_{0}.\]

For $\delta=2,3$, we know that $T=C_{0}\cup C_{1}$ and $T^{-1}=C_{0}\cup C_{2^{m-1}-1}$. Note that $2^{m-1}-1$ is the largest coset leader modulo $2^m-1$, then \[T^{\perp}=\Z_{n}\backslash T^{-1}=C_{1}\cup C_{2}\cup \cdots \cup C_{2^{m-1}-2}.\] Thus, $\widetilde{\C}_{(2,n,\delta)}^{\perp}$ is a BCH code with $b=1$ and designed distance $\delta=2^{m-1}-1$.

For $2^{m-1}-2^{\lfloor \frac{m-1}{2} \rfloor}\leq \delta \leq n$, we divide it into two cases:
\begin{itemize}
\item If $2^{m-1}\leq \delta \leq n$, then $T=\Z_{n}$, which leads to $\widetilde{\C}_{(2,n,\delta)}^{\perp}=\lbrace\textbf{0}\rbrace$.

\item If $2^{m-1}-2^{\lfloor \frac{m-1}{2} \rfloor} \leq \delta \leq 2^{m-1}-1$, note that $\delta_{2,n}+1=2^{m-1}-2^{\lfloor \frac{m-1}{2} \rfloor}$. Then
	\begin{center}
			$T^{\perp}=(\Z_{n}\backslash T)^{-1}=(\Z_{n}\backslash (C_{0}\cup C_{1}\cup C_{2}\cup \cdots \cup C_{\delta_{2,n}}))^{-1}=(C_{\delta_{1,n}})^{-1}=C_{CL({n-\delta_{1,n}})}=C_{1}$.
	\end{center}
Obviously, $\widetilde{\C}_{(2,n,\delta)}^{\perp}$ is a narrow-sense BCH code.
\end{itemize}

For $3<\delta <2^{m-1}-2^{\lfloor \frac{m-1}{2} \rfloor}$, we claim  that $\widetilde{\C}^{\perp}_{(2,n,\delta)}$ is not a narrow-sense BCH code.

 By Lemma \ref{lem2}, we know that $\delta_{1,n}=2^{m-1}-1>2^{m-1}-2^{\lfloor \frac{m-1}{2} \rfloor}>\delta$, then $\delta_{1,n}\notin T$. Note that $CL(n-\delta_{1,n})=1$, i.e., $C_{1}\in T^{\perp}$. If $\widetilde{\C}_{(2,n,\delta)}^{\perp}$ is a narrow-sense BCH code, then there must exist an integer $r\geq 1$ such that $T^{\perp}=C_{1}\cup C_{2}\cup \cdots \cup C_{r-1}$.

  By Lemma \ref{lem7}, we know that $\C_{(2,n,\delta)}$ is not a dually-BCH code for all any $3<\delta <2^{m-1}-2^{\lfloor \frac{m-1}{2} \rfloor}$, which means there is no $r\geq 1$ such that the defining set of $\C^{\perp}_{(2,n,\delta)}$ is equal to $C_{0}\cup C_{1}\cup\cdots\cup C_{r-1}$. Note that  $T^{\perp}=((\Z_{n}\backslash (T^{'}))^{-1} )\backslash C_{0}$, then there is no integer $r\geq 1$ such that $T^{\perp}=C_{1}\cup C_{2}\cup \cdots \cup C_{r-1}$. i.e., $\widetilde{\C}^{\perp}_{(2,n,\delta)}$ is not a narrow-sense BCH code.
Thus we complete the proof.
\end{proof}

\begin{theorem}\label{th15}
	Let $q\geq 3$ and $m\geq 2$. Then $\widetilde{\C}_{(2,n,\delta)}^{\perp}$ is a narrow-sense BCH code if and only if
	\begin{center}
		$\delta=2$ or $(q-1)q^{m-1}-q^{\lfloor \frac{m-1}{2} \rfloor}\leq \delta \leq (q-1)q^{m-1}-1$.
	\end{center}
\end{theorem}
	
\begin{proof}
	The proof is very similar to that of Theorem \ref{th14}, hence we omit it.
	\end{proof}

\section{The case of $n=\frac{q^{2s}-1}{q+1}$}\label{set4}
In this section, we always assume $m=2s$, $n=\frac{q^{m}-1}{q+1}$ and $\beta_{1}=\alpha^{q+1}$.
\subsection{The computation of $\delta_{2,n}$}
For $s=2$ and  $2\nmid q$, $\delta_{2,n}$ has been determined in \cite{RefJ19}. 
Next we give $\delta_{2,n}$ for $s\geq 3$.

\begin{lemma}\label{lem16}
	Let $q$ be a prime power, then
	\begin{center}
		$\delta_{2,n}$ =$\begin{cases}
			\frac{(q-1)q^{2s-1}-q^{s+1}-1}{q+1},& \text{if $2\nmid s$ and $s>4$}; \\
			\frac{(q-1)q^{2s-1}-q^{s+2}-1}{q+1},& \text{if $2\mid s$ and $s>6$},
		\end{cases}$
	\end{center}
	and $\mid C_{\delta_{2,n}}\mid =2s$.
	\begin{proof}
		If $2\nmid s$, we have $(q+1)\mid(q^{s+1}-1)$ and
		\begin{center}
			$(q+1)\mid(q^{2s}-1)-(q^{2s-1}+1)-(q^{s+1}-1)=(q-1)q^{2s-1}-q^{s+1}-1$.
		\end{center}
		Note that by Corollary \ref{cor10}, $\delta_{4,q^m-1}=(q-1)q^{2s-1}-q^{s+1}-1$, this implies that $\frac{(q-1)q^{2s-1}-q^{s+1}-1}{q+1}\in \MinRep_n$ by Lemma \ref{lem3}.
		
We claim that $\delta_{2,n}=\frac{(q-1)q^{2s-1}-q^{s+1}-1}{q+1}$.
		Suppose there exists an integer $a$ such that $a\in\MinRep_n$ and $\delta_{2,n}<a<\delta_{1,n}$. So $a(q+1)\in\MinRep_{q^m-1}$, $\delta_{4,q^m-1}=\delta_{2,n}(q+1)<a(q+1)<\delta_{2,q^m-1}=\delta_{1,n}(q+1)$. This means \begin{equation}\label{eq:q}a(q+1)=\delta_{3,q^m-1}=(q-1)q^{2s-1}-1-q^{s}.\end{equation}
Note that $\gcd(q+1,(q-1)q^{2s-1}-1-q^{s})=\gcd(q+1,q-1)<q+1$, then $(q+1)\nmid(q-1)q^{2s-1}-1-q^{s}$. By Eq.~\eqref{eq:q}, this is impossible.
Thus, $\delta_{2,n}=\frac{(q-1)q^{2s-1}-q^{s+1}-1}{q+1}$.
		
		Let $\mid C_{\delta_{2,n}}\mid=l$, then
		\begin{center}
			$\frac{q^{2s}-1}{q+1}\mid \big( \frac{(q-1)q^{2s-1}-q^{s+1}-1}{q+1}\big)(q^{l}-1)$   $\Longleftrightarrow$  $(q^{2s}-1)\mid \big((q-1)q^{2s-1}-q^{s+1}-1\big)(q^{l}-1)$.
		\end{center}
		Thus $l=2s$. The case $2\mid s$ is similar. 
		\\Thus we complete the proof.
	\end{proof}	
\end{lemma}
\begin{lemma}\label{lem17}
	Let $2\nmid q$.
	\begin{itemize}
		\item [(1)]If $s=3$, then $\delta_{2,n}=\frac{(q-1)q^{5}-q^{4}-1}{q+1}$ and $\mid C_{\delta_{2,n}}\mid=6$.
		\item [(2)]If $s=4$, then $\delta_{2,n}=\frac{(q-1)q^{7}-q^{6}-1}{q+1}$ and $\mid C_{\delta_{2,n}}\mid=8$.

		\item [(3)]If $s=6$, then $\delta_{2,n}=\frac{(q-1)q^{11}-q^{7}-1}{q+1}$ and $\mid C_{\delta_{2,n}}\mid=12$.
	\end{itemize}
	\begin{proof}
		We just give the proof for Case (2), since the proofs for the other cases are similar.
		
		Note that $[((q-1)q^{7}-q^{6}-1)q^{j}]_{q^{8}-1}> (q-1)q^{7}-q^{6}-1$ for any $1\leq j\leq 7$, this implies that $(q-1)q^{7}-q^{6}-1\in \MinRep_{q^{8}-1}$ and $\mid C_{(q-1)q^{7}-q^{6}-1}\mid=8$. Since $(q+1)\mid(q-1)q^{7}-q^{6}-1$, this implies that $\frac{(q-1)q^{7}-q^{6}-1}{q+1}\in \MinRep_n$ and $\mid C_{\frac{(q-1)q^{7}-q^{6}-1}{q+1}}\mid=8$ by Lemma \ref{lem3}.
		
		We claim that $\delta_{2,n}=\frac{(q-1)q^{7}-q^{6}-1}{q+1}$.
		Suppose there exists an integer $a$ such that $a\in\MinRep_n$ and $\delta_{2,n}<a<\delta_{1,n}$, then $a(q+1)\in\MinRep_{q^8-1}$.
		
		Note that all coset leaders modulo $q^{8}-1$ between $\delta_{1,n}(q+1)=(q-1)q^{7}-q^{4}-1$ and $\delta_{2,n}(q+1)=(q-1)q^{7}-q^{6}-1$ are $ (q-1)q^{7}-q^{5}-q^{2}-1$,  $(q-1)q^{7}-q^{5}-q^{3}-1$ and $(q-1)q^{7}-q^{5}-q^{3}-q-1$, this means
		\begin{equation}\label{equ:6}
			(q+1)a=(q-1)q^{7}-q^{5}-q^{2}-1 \text{, } (q-1)q^{7}-q^{5}-q^{3}-1  \,\, or \,\, (q-1)q^{7}-q^{5}-q^{3}-q-1.
		\end{equation}
	
		Note that $(q+1)\nmid(q-1)q^{7}-q^{5}-q^{2}-1$, $(q+1)\nmid(q-1)q^{7}-q^{5}-q^{3}-1$ and  $(q+1)\nmid (q-1)q^{7}-q^{5}-q^{3}-q-1$, this is impossible by Eq. (\ref{equ:6}).
		Therefore, $\delta_{2,n}=\frac{(q-1)q^{7}-q^{6}-1}{q+1}$.
		Thus we complete the proof.
	\end{proof}	
\end{lemma}
\subsection{BCH Codes and Dually-BCH Codes}
\begin{theorem}\label{th18} Let $q$ be a prime power and $s>4$. 
	 For $ 2\nmid s$,
	\begin{itemize}
		\item If $\delta_{2,n}+1\leq \delta \leq \delta_{1,n}$, then $\C_{(q,n,\delta)}$ has parameters $[\frac{q^{2s}-1}{q+1},s+1,d(\C_{(q,n,\delta)})\geq \delta]$ and its dual $\C^{\perp}_{(q,n,\delta)}$ has parameters $[\frac{q^{2s}-1}{q+1},\frac{q^{2s}-1}{q+1}-s-1,2]$.
		\item If $\delta=\delta_{2,n}$, then $\C_{(q,n,\delta_{2,n})}$ has parameters $[\frac{q^{2s}-1}{q+1},3s+1,d(\C_{(q,n,\delta_{2,n})})\geq \delta_{2,n}]$ and its dual $\C^{\perp}_{(q,n,\delta_{2,n})}$ has parameters
		$[\frac{q^{2s}-1}{q+1},\frac{q^{2s}-1}{q+1}-3s-1, 3\leq d(\C^{\perp}_{(q,n,\delta_{2,n})})\leq 4].$
	\end{itemize}
	For $2\mid s$ and $s\neq6$,
	\begin{itemize}
		\item If $\delta_{2,n}+1\leq \delta \leq \delta_{1,n}$, then $\C_{(q,n,\delta)}$  has parameters $[\frac{q^{2s}-1}{q+1},2s+1,d(\C_{(q,n,\delta)})\geq \delta]$ and its dual $\C^{\perp}_{(q,n,\delta)}$ has parameters $[\frac{q^{2s}-1}{q+1},\frac{q^{2s}-1}{q+1}-2s-1,3\leq d(\C^{\perp}_{(q,n,\delta)}) \leq 4]$.
		\item If $\delta=\delta_{2,n}$, then $\C_{(q,n,\delta_{2,n})}$ has parameters
		$[\frac{q^{2s}-1}{q+1},4s+1,d(\C_{(q,n,\delta_{2,n})})\geq \delta_{2,n}]$ and its dual $\C^{\perp}_{(q,n,\delta_{2,n})}$ has parameters
		$[\frac{q^{2s}-1}{q+1},\frac{q^{2s}-1}{q+1}-4s-1,3\leq d(\C^{\perp}_{(q,n,\delta_{2,n})})\leq 4]$.
		
	\end{itemize}
\end{theorem}

\begin{proof}
	Form Lemmas \ref{lem3} and \ref{lem16}, the parameters of $\C_{(q,n,\delta)}$ can be obtained directly. Next we consider the parameters of $\C^{\perp}_{(q,n,\delta)}$.

 Case 1. For $\delta_{2,n}+1\leq \delta \leq \delta_{1,n}$, by definition, we know the defining set of $\C^{\perp}_{(q,n,\delta)}$ with respect to $\beta_{1}$ is $T^{\perp}=(\Z_{n}\backslash T)^{-1}=C_{0}\cup C_{CL(n-\delta_{1,n})}$. Then $d(\C^{\perp}_{(q,n,\delta)})\geq2$. 
	\begin{itemize}
	\item [1.1)] If $2\nmid s$, then \[n-\delta_{1,n}=\frac{q^{2s}-1-((q-1)q^{2s-1}-q^{s-1}-1)}{q+1}=\frac{q^{2s-1}+q^{s-1}}{q+1}.\]
	Thus $CL(n-\delta_{1,n})=\frac{q^s+1}{q+1}$ and $\mid C_{\frac{q^s+1}{q+1}}\mid=s$. This means
	\begin{center}
		$\dim(\C^{\perp}_{(q,n,\delta)})=n-|T^{\perp}|=n-s-1.$
	\end{center}
	Form the sphere packing bound, 
	\begin{equation}\label{con:e6}
		q^{n-(s+1)}\sum_{i=0}^{\lfloor \frac{d-1}{2}\rfloor}(q-1)^{i} \left(\mathop{}_{i}^{n}\right)\leq q^{n} \Longrightarrow \sum_{i=0}^{\lfloor \frac{d-1}{2}\rfloor}(q-1)^{i} \left(\mathop{}_{i}^{n}\right)\leq q^{s+1}.
	\end{equation}
	Suppose $d= 3$, then by Eq.~\eqref{con:e6}, we have
	\[1+(q-1)\frac{q^{2s}-1}{q+1}\leq q^{s+1},\] which is impossible due to $s>4$. Thus $d=2$.
	\item [1.2)] If $2\mid s$ and $s\neq6$, then 
	\[n-\delta_{1,n}=\frac{q^{2s}-1-((q-1)q^{2s-1}-q^{s}-1)}{q+1}=\frac{q^{2s-1}+q^{s}}{q+1}.\]
	Thus $CL(n-\delta_{1,n})=\frac{q^{s-1}+1}{q+1}$ and $\mid C_{\frac{q^{s-1}+1}{q+1}}\mid=2s$. This means
	\begin{center}
		$\dim(\C^{\perp}_{(q,n,\delta)})=n-|T^{\perp}|=n-2s-1$.
	\end{center}
	Suppose $d({\C^{\perp}_{(q,n,\delta)}})=5$, by Eq.~\eqref{con:e6}, 	we know that
	\begin{equation}\label{eq:b1}1+(q-1)\frac{q^{2s}-1}{q+1}+(q-1)^2\frac{q^{2s}-1}{q+1}
		\frac{q^{2s}-q-2}{q+1}>q^{2s+1}.\end{equation}
	Thus $2\leq d({\C^{\perp}_{(q,n,\delta)}})\leq 4$. Set $s_{1}=\frac{q^{s-1}+1}{q+1}$, then  
	the parity-check matrix of $\C^{\perp}_{(q,n,\delta)}$ is
	\begin{center}
		$H$=$\begin{pmatrix}
			1 & 1 & 1 & \cdots & 1 \\
			1 & \beta_{1}^{s_{1}} & \beta_{1}^{2s_{1}} & \cdots &\beta_{1}^{(n-1)s_{1}}
		\end{pmatrix}.$
	\end{center}
	We claim that any two columns are linear independent over $\F_q$. Suppose there exists $(c_1,c_2)\in\F_q\times\F_q\setminus\{(0,0)\}$ such that $c_1+c_2=0$ and $c_1\beta_1^{is_1}+c_2\beta_1^{js_1}=0$ for any $0 \leq i < j\leq n-1$, which equals to
	\begin{equation}\label{con:e9}
		\beta_{1}^{s_{1}i}=\beta_{1}^{s_{1}j} \Longleftrightarrow\alpha^{(j-i)(q^{s-1}+1)}=1\Longleftrightarrow(j-i)(q^{s-1}+1)\equiv0\pmod{q^{2s}-1}.
	\end{equation}
	Since $\gcd(2s,s-1)=\gcd(s,s-1)=1$, then $\gcd(q^{s-1}+1,q^{2s}-1)=q^{\gcd(s-1,2s)}+1=q+1$.  Eq. (\ref{con:e9}) holds if and only if $i-j\equiv0\pmod{\frac{q^{2s}-1}{q+1}}$, this means $i=j$. Thus $d({\C^{\perp}_{(q,n,\delta)}})\geq 3$.
    \end{itemize} 
	Case 2. For $\delta=\delta_{2,n}$, the defining set of $\C^{\perp}_{(q,n,\delta_{2,n})}$ is $T^{\perp}=(\Z_{n}\backslash T)^{-1}=C_{0}\cup C_{CL(n-\delta_{1,n})}\cup  C_{CL(n-\delta_{2,n})}$.
	\begin{itemize}
	\item [2.1)]If $2\nmid s$, then \[n-\delta_{2,n}=\frac{q^{2s}-1-((q-1)q^{2s-1}-q^{s+1}-1)}{q+1}=\frac{q^{2s-1}+q^{s+1}}{q+1}.\]
	Clearly, $CL(n-\delta_{2,n})=\frac{q^{s-2}+1}{q+1}$ and $\mid C_{\frac{q^{s-2}+1}{q+1}}\mid=2s$. We have
	\begin{center}
		$\dim(\C^{\perp}_{(q,n,\delta_{2})})=n-|T^{\perp}|=n-3s-1$.
	\end{center}
	
	By Eq.~\eqref{eq:b1}, we know $2\leq d(\C^{\perp}_{(q,n,\delta_{2,n})})\leq 4$.
	Take $s_{1}=\frac{q^s+1}{q+1}$ and $s_{2}=\frac{q^{s-2}+1}{q+1}$, then the parity-check matrix of $\C^{\perp}_{(q,n,\delta_{2})}$ is
	\begin{center}
		$H$=$\begin{pmatrix}
			1 & 1 & 1 & \cdots & 1 \\
			1 & \beta_{1}^{s_{1}} & \beta_{1}^{2s_{1}} & \cdots &\beta_{1}^{(n-1)s_{1}}\\
			1 & \beta_{1}^{s_{2}} & \beta_{1}^{2s_{2}} & \cdots &\beta_{1}^{(n-1)s_{2}}\\
		\end{pmatrix}.$
	\end{center}
	Note that $d(\C^{\perp}_{(q,n,\delta_{2,n})})=2$ if and only if there exist $i,j$ such that $0\leq i\neq j\leq n-1$ and
	\begin{equation}
		\begin{cases}
			\beta_{1}^{is_{1}}=\beta_{1}^{js_{1}}\\
			\beta_{1}^{is_{2}}=\beta_{1}^{js_{2}}
		\end{cases}
		\Longrightarrow
		\begin{cases}
			\beta_{1}^{(j-i)s_{1}}=1\\
			\beta_{1}^{(j-i)s_{2}}=1
		\end{cases}.	\label{con:e7}
	\end{equation} Eq.~\eqref{con:e7} equals to \begin{align}
		(i-j)(q^s+1)\equiv0\pmod{q^{2s}-1}\label{eq:9}\ \text{and} \\ 
		(i-j)(q^{s-2}+1)\equiv0\pmod{q^{2s}-1}\label{eq:8}.
	\end{align}
	By ~\eqref{eq:9}, we know that $i-j\equiv0\pmod{q^s-1}$
	where $0\leq i< j\leq n-1$. Thus Eq.~\eqref{eq:8} holds if and only if there exists an integer $k$ such that $0\leq k\leq \frac{q^s+1}{q+1}-1$ and \begin{equation}\label{eq:7}k(q^{s-2}+1)\equiv0\pmod{q^s+1}.\end{equation} Note that $2\nmid (s-2)$, then $\gcd(q^{s}+1,q^{s-2}+1)=q+1,$ i.e., Eq.~\eqref{eq:7} holds if and only if $k\equiv0\pmod{\frac{q^s+1}{q+1}}$, this is impossible. Thus $3\leq d(\C^{\perp}_{(q,n,\delta_{2,n})})\leq4$. 
	\item [2.2)]The case $2\mid s$ and $s\neq6$ is similar to {2.1)}.
\end{itemize}
	Thus we complete the proof.
\end{proof}
\begin{theorem}\label{th19}
	Let $q$ be a prime power and $2\nmid q$, then we have the following.
	\begin{itemize}
		\item If $s=3$, then $\C_{(q,n,\delta)}$ with $\delta_{2,n}+1\leq \delta \leq \delta_{1,n}$ has parameters $[\frac{q^{6}-1}{q+1},4,d\geq \delta]$ and $\C_{(q,n,\delta_{2,n})}$ has parameters $[\frac{q^{6}-1}{q+1},10,d\geq \delta_{2,n}]$.
		\item If $s=4$, then $\C_{(q,n,\delta)}$ with $\delta_{2,n}+1\leq \delta \leq \delta_{1,n}$ has parameters $[\frac{q^{8}-1}{q+1},9,d\geq \delta]$ and $\C_{(q,n,\delta_{2,n})}$ has parameters
 $[\frac{q^{8}-1}{q+1},17,d\geq \delta_{2,n}]$.
	\item If $s=6$, then $\C_{(q,n,\delta)}$ with $\delta_{2,n}+1\leq \delta \leq \delta_{1,n}$ has parameters $[\frac{q^{12}-1}{q+1},13,d\geq \delta]$ and $\C_{(q,n,\delta_{2,n})}$ has parameters $[\frac{q^{12}-1}{q+1},25,d\geq \delta_{2,n}]$.
	\end{itemize}
	\begin{proof}
	Form Lemmas \ref{lem3} and \ref{lem17}, we can obtain the results directly.
	\end{proof}
\end{theorem}

\begin{theorem}\label{th20}
	Let $q$ be a prime power and $s>4$, then we have the following.
	\begin{itemize}
		\item If $2\nmid s$, then the code $\widetilde{\C}^{\perp}_{(q,n,\delta)}$ with $\delta_{2,n}+1\leq \delta \leq \delta_{1,n}$ has parameters
			\begin{center}
				$[\frac{q^{2s}-1}{q+1},\frac{q^{2s}-1}{q+1}-s,2]$.
			\end{center}
		\item If $2\mid s$ and $s>6$, then the code $\widetilde{\C}^{\perp}_{(q,n,\delta)}$ with $\delta_{2,n}+1\leq \delta \leq \delta_{1,n}$ has parameters
			\begin{center}
				$[\frac{q^{2s}-1}{q+1},\frac{q^{2s}-1}{q+1}-2s, 3]$.
			\end{center}
	\end{itemize}
	\begin{proof}
		By definition, we know that the defining set of $\widetilde{\C}^{\perp}_{(q,n,\delta)}$ with respect to $\beta_{1}$ is $T^{\perp}=(\Z_{n}\backslash T)^{-1}=C_{CL(n-\delta_{1,n})}$, then $d(\widetilde{\C}_{(q,n,\delta)})\geq 2$. Note that $\widetilde{\C}_{(q,n,\delta)}\subset {\C}_{(q,n,\delta)} $, which means $ \widetilde{\C}^{\perp}_{(q,n,\delta)} \supset {\C}^{\perp}_{(q,n,\delta)}$, then $d({\widetilde{\C}^{\perp}_{(q,n,\delta)}})\leq d({{\C}^{\perp}_{(q,n,\delta)}})$.
		 
		If $2\nmid s$, we have $2\leq d({\widetilde{\C}^{\perp}_{(q,n,\delta)}})\leq d({\C}^{\perp}_{(q,n,\delta)})=2 \Longrightarrow d({\widetilde{\C}^{\perp}_{(q,n,\delta)}})=2$. Since $CL(n-\delta_{1,n})=\frac{q^s+1}{q+1}$ and $\mid C_{\frac{q^s+1}{q+1}}\mid=s$. This means
		\begin{center}
			$\dim(\widetilde{\C}^{\perp}_{(q,n,\delta)})=n-\mid C_{CL(n-\delta_{1})}\mid=n-s.$
		\end{center}
	Then the parameters of code $\widetilde{\C}^{\perp}_{(q,n,\delta)}$ are determined.
	
		If $2\mid s$ and $s>6$, note that $CL(n-\delta_{1,n})=\frac{q^{s-1}+1}{q+1}$ and $\mid C_{\frac{q^{s-1}+1}{q+1}}\mid=2s$. This means
		\begin{center}
			$\dim(\widetilde{\C}^{\perp}_{(q,n,\delta)})=n-\mid C_{CL(n-\delta_{1})}\mid=n-2s.$
		\end{center}
		Set $s_{1}=\frac{q^{s-1}+1}{q+1}$, then the parity-check matrix is
		\begin{center}
			$H$=$\begin{pmatrix}
				1 & \beta_{1}^{s_{1}} & \beta_{1}^{2s_{1}} & \cdots &\beta_{1}^{(n-1)s_{1}}
			\end{pmatrix}.$
		\end{center}
		Since $\gcd(q^{s-1}+1,q^{2s}-1)=q^{\gcd(s-1,2s)}+1=q+1$, then $\gcd(s_{1},n)=1$. Therefore, $\beta_{1}^{s_{1}}$ is a primitive $n$-th root of unity, then $\widetilde{\C}^{\perp}_{(q,n,\delta)}$ is the Hamming code and  $d(\widetilde{\C}^{\perp}_{(q,n,\delta)})=3$.
		Then the parameters of code $\widetilde{\C}^{\perp}_{(q,n,\delta)}$ are determined.
		Thus we complete the proof.
	\end{proof}
\end{theorem}
\begin{example}
	Let $(q,s)=(2,5)$, we have $n=341$ then $\delta_{1,n}=165$ and $\delta_{2,n}=149$. The BCH code $\C_{(2,341,\delta)}$ with $149\leq \delta\leq 165$ has parameters $[341,6,d\geq \delta]$, the BCH code $\C_{(2,341,149)}$ has parameters $[341,16,d\geq 149]$, the code ${\C}^{\perp}_{(2,341,149)}$ has parameters $[341,325,4\geq d\geq 3]$.
\end{example}
\begin{example}
	Let $(q,s)=(3,3)$, we have $n=182$. Then $\delta_{2,n}=101$. The BCH code $\C_{(3,182,101)}$ has parameters $[182,10,\geq 101]$.
\end{example}

For $2\mid s$ and $s\geq 4$, we will investigate the dimension of the code $\C_{(q,n,\delta)}$, where $\lceil \frac{q}{2}\rceil q^{s-1}\leq \delta \leq \frac{q^{s+1}+1}{q+1}$.
\begin{lemma}\label{lem23}
	Let $\lceil \frac{q}{2}\rceil q^{s-1}\leq i< \frac{q^{s+1}+1}{q+1}$ and $[i]_{q}\neq 0$, then we have the following.
	\begin{itemize}
		\item [(1)]If $ 2\mid q$ and let $i=\frac{a_{s}q^{s}+a_{0}}{q+1}$, where $a_{s}=q-t$ and $a_{0}=t+1$ for all $1\leq t\leq \frac{q-2}{2}$, then $i$ is not a coset leader. Otherwise, $i$ is a coset leader and $\mid C_{i}\mid=2s$.
		\item [(2)]If $2\nmid q$ and let $i=\frac{a_{s}q^{s}+a_{0}}{q+1}$, where $a_{s}=q-t$ and $a_{0}=t+1$ for all $1\leq t\leq \frac{q-3}{2}$, then $i$ is not a coset leader. Otherwise, $i$ is a coset leader and $\mid C_{i}\mid=2s$.
	\end{itemize}
	\begin{proof}
		We just give the proof for Case	(1), since the proof for Case (2) is similar.
		
		Note that $i\in\MinRep_n$ if and only if $i(q+1)\in\MinRep_{q^{m}-1}$
		by Lemma \ref{lem3}, then we need to find $a$ such that $a\in\MinRep_{q^{m}-1}$, $\frac{q^{s}(q+1)}{2}\leq a \leq q^{s+1}$ and $(q+1)\mid a$.
		
		Note that $[\frac{q^{s}(q+1)}{2},q^{s+1}]\subset [q^{s}+1,q^{s+1}]$, then we divide $\frac{q^{s}(q+1)}{2}\leq a \leq q^{s+1}$ into the following three cases by Lemma \ref{lem2}.
		
		\begin{itemize}
			\item [1)] $a=c(q^{s}+1)$ for $1\leq c\leq q-1$. Since $2\mid q$ and $2\mid s$, then $\gcd(q^{s}+1,q^{2}-1)=1$, which implies that $\gcd(q^{s}+1,q+1)=1$. Hence, $q+1\nmid c(q^{s}+1)$ for all $1\leq c\leq q-1$, then we do not consider these values.
			\item [2)] $a=a_{s}q^{s}+a_{0}$ for $1\leq a_{0}<a_{s}\leq q-1$.
			Since $a=a_{s}(q^{s}-1)+a_{s}+a_{0}$, then $q+1\mid a$ if 
			$\begin{cases}
				a_{s}=q-t \\a_{0}=t+1
			\end{cases}$
			for all $1\leq t\leq \frac{q}{2}-1$. Note that $$a_{s}q^{s}+a_{0}\geq (\frac{q}{2}+1)q^{s}+\frac{q}{2}>\frac{q^{s}(q+1)}{2},$$
			then $q+1\mid a$ and $a\in [\frac{q^{s}(q+1)}{2},q^{s+1}]$ if 
			$\begin{cases}
				a_{s}=q-t \\a_{0}=t+1
			\end{cases}$
			for all $1\leq t\leq \frac{q}{2}-1$, but $a\notin\MinRep_{q^{m}-1}$.
			\item [3)] Otherwise, for any $a\in [\frac{q^{s}(q+1)}{2},q^{s+1}]$ and $q+1\mid a$, we know that $a\in\MinRep_{q^{m}-1}$ and $\mid C_{a}\mid =2s$.
		\end{itemize}

	For $\frac{q^{s}(q+1)}{2}\leq a< q^{s+1}+1$, we konw that $a\notin\MinRep_{q^{m}-1}$ and $(q+1)\mid a$ if and only if $a=a_{s}q^{s}+a_{0}$ and 
	$\begin{cases}
		a_{s}=q-t \\a_{0}=t+1
	\end{cases}$
	for all $1\leq t\leq \frac{q}{2}-1$. Combining all the cases above, the desired conclusion then follows.
	\end{proof}
\end{lemma}

\begin{theorem}\label{th24}
	Let $2\mid s$ and $s\geq 4$. Then the dimension $k$ of $\C_{(q,n,\delta)}$ is given as follows, where $\lceil \frac{q}{2}\rceil q^{s-1}+1\leq \delta \leq \frac{q^{s+1}+1}{q+1}$.
	\begin{itemize}
		\item [(1)]If $2\mid q$, then we have
		$$k=n-2s(\delta+\frac{q-2}{2}-\left\lfloor\frac{(\delta-1)(q+1)}{q^{s}}\right\rfloor)+2s\left\lfloor\frac{\delta-1}{q}\right\rfloor.$$
	\item [(2)]If $2\nmid q$, then we have
	$$k=n-2s(\delta+\frac{q-1}{2}-\left\lfloor\frac{(\delta-1)(q+1)}{q^{s}}\right\rfloor)+2s\left\lfloor\frac{\delta-1}{q}\right\rfloor+s.$$
	\end{itemize}
 \begin{proof}
 	Form Lemmas \ref{lem6} and \ref{lem23}, we can obtain the results directly.
 \end{proof}
\end{theorem}
\begin{example}
	Let $(q,s)=(2,4)$, we have $n=85$. The BCH code $\C_{(2,85,9)}$ has parameters $[85,53,\geq 9]$. The code is almost optimal in the sense that the minimum distance of the optimal binary linear code with length 85 and dimension 53 is 10 according to the tables of best codes known in (\cite{RefJ11}) when the equality holds.
\end{example}
We will give the dimension and the minimum distance for special designed distance in the following theorem.
\begin{theorem}\label{th26}
	Let $a$ be an integer and $1\leq a\leq q-1$. Then the following holds.
	\begin{itemize}
		\item [(1)]For $2\nmid s$ and $s\geq 3$ $(s\geq 5$ if $q=2)$,
		\begin{itemize}
			\item [(1.1)]if $\delta=a\frac{q^{s}-1}{q-1}$, then $\C_{(q,n,\delta)}$ has parameters $\left[\frac{q^{2s}-1}{q+1},k,\delta\right]$, where
			\begin{center}
				$k$=$\begin{cases}
					n-2s(\delta-1)+2s\left\lfloor\frac{\delta-1}{q}\right\rfloor+s\left\lfloor\frac{(\delta-1)(q+1)}{q^{s}+1}\right\rfloor,& \text{if $1\leq a\leq q-3$} ;\\
					n-2s(\delta-1)+2s\left\lfloor\frac{\delta-1}{q}\right\rfloor+s(q+1),& \text{if $a=q-2$};\\
					n-s\left((q-1)(2q^{s-1}-3)-2\right),&\text{if $a=q-1$}.
				\end{cases}$
			\end{center}
			\item [(1.2)]if $\delta=a\frac{q^{s}+1}{q+1}$, then $\C_{(q,n,\delta)}$ has parameters $\left[\frac{q^{2s}-1}{q+1},k,\delta\right]$, where
			\begin{center}
				$k$=$\begin{cases}
					n-2s(\delta-1)+2s\left\lfloor\frac{\delta-1}{q}\right\rfloor,& \text{if $a=1$} ;\\
					n-s(2\delta-a-1)+2s\left\lfloor\frac{\delta-1}{q}\right\rfloor,& \text{if $2\leq a\leq q-1$}.
				\end{cases}$	
			\end{center}
		\end{itemize}
		\item [(2)]For $2\mid s$ and $s\geq 4$, if $\delta=a\frac{q^{s}-1}{q^{2}-1}$, then $\C_{(q,n,\delta)}$ has parameters $\left[\frac{q^{2s}-1}{q+1},k,\delta\right]$, where
		\begin{center}
			$k$=$n-2s(\delta-1)+2s\left\lfloor \frac{\delta-1}{q} \right\rfloor$.
		\end{center}
	\end{itemize}
	\begin{proof}
	We only prove Case (1), since the proof for Case (2) is similar.
		
	Note that $q+1\mid q^{s}+1$ and $q-1\mid q^{s}-1$ if $2\nmid s$, then $\frac{q^{s}-1}{q-1}$ and $\frac{q^{s}+1}{q+1}$ are integers. Since $q^{s}-1\mid \frac{q^{s}+1}{q+1}(q^{s}-1)$ and $ q^{s}+1\mid \frac{q^{s}-1}{q-1}(q^{s}+1)$, then $\frac{q^{s}-1}{q-1}\mid \frac{n}{q-1}$ and $\frac{q^{s}+1}{q+1}\mid \frac{n}{q-1}$.
	Obviously, $q^{2}-1\mid q^{2s}-1$, then $\gcd(n,q-1)=q-1$ and $\gcd(n,q)=1$. Hence, if $\delta=a\frac{q^{s}-1}{q-1}$ and $a\frac{q^{s}+1}{q+1}$, then we have $d(\C_{(q,n,\delta)})=\delta$ by Lemma \ref{lem5}.
		
		For $\delta=a\frac{q^{s}-1}{q-1}$, we divide into the following two cases by Lemma \ref{lem6}.
		\begin{itemize}
			\item [1)]For $q=2$ and $s\geq 5$, then $a=1$ and $\delta=\frac{2^{s}-1}{2-1}$. Note that, $(2-1)2^{s-1}+\frac{2^{s}+1}{2+1}+1\leq \delta\leq 2^{s}+1$, and $\dim(\C_{(2,n,\delta)})=n-s(2^{s}-5)$.
			\item [2)]For $q\geq 3$ and $s\geq 3$,
			\begin{itemize}
				\item[$\bullet$] if $1\leq a\leq q-3$, note that $$\frac{q^{s}+1}{q+1}+1\leq \frac{q^{s}-1}{q-1}\leq \delta\leq (q-3)\frac{q^{s}-1}{q-1} \leq (q-1)\frac{q^{s}+1}{q+1}+1,$$ and $\dim(\C_{(q,n,\delta)})=n-2s(\delta-1)+2s\left\lfloor\frac{\delta-1}{q}\right\rfloor+s\left\lfloor\frac{(\delta-1)(q+1)}{q^{s}+1}\right\rfloor$.
				\item[$\bullet$] if $a=q-2$, note that $$(q-1)\frac{q^{s}+1}{q+1}+2\leq (q-2)\frac{q^{s}-1}{q-1} \leq \frac{q^{s+1}-1}{q+1}+2,$$ and $\dim(\C_{(q,n,\delta)})=n-2s(\delta-1)+2s\left\lfloor\frac{\delta-1}{q}\right\rfloor+s(q+1)$.
				\item[$\bullet$] if $a=q-1$, note that $$(q-1)q^{s-1}+\frac{q^{s}+1}{q+1}+1\leq (q-1)\frac{q^{s}-1}{q-1} \leq q^{s}+1,$$ and $\dim(\C_{(q,n,q^{s}-1)})=n-s\left((q-1)(2q^{s-1}-3)-2\right)$.
			\end{itemize}
		\end{itemize}
	
		For $\delta=a\frac{q^{s}+1}{q+1}$, we still obtain the dimension by determining the range of $\delta$, i.e., the result corresponds to (1.2). Thus we complete the proof.
	\end{proof}
\end{theorem}
\begin{example}
	We have the following examples for the code of Theorem \ref{th26}.
	\\$\bullet$ Let $(q,s)=(2,4)$, we have $n=85$. The BCH code $\C_{(2,85,69)}$ has parameters $[85,69,5]$. The code is almost optimal according to the tables of best codes known in (\cite{RefJ11}).
	\\$\bullet$ Let $(q,s)=(2,5)$, we have $n=341$. The BCH code $\C_{(2,341,206)}$ has parameters $[341,206,31]$. The BCH code $\C_{(2,341,291)}$ has parameters $[341,291,11]$.
\end{example}

For $b=1$, the conditions of $\C_{(q,n,\delta)}$ being dually-BCH codes have been given by Lemma~\ref{lem8}. For $b=2$, we have

\begin{theorem}\label{th28}
	Let $q\geq 3$ and $s\geq 2$. Then the following statements hold.
	\begin{itemize}
		\item [(1)]If $s=2$, then $\C_{(q,n,\delta,2)}$ is a dually-BCH code if and only if
		\begin{center}
			$\delta_{1,n}-1\leq \delta\leq n-1$.
		\end{center}
		\item [(2)]If $s\neq2$, then $\C_{(q,n,\delta,2)}$ is a dually-BCH code if and only if
		\begin{center}
			$\delta_{1,n}\leq \delta\leq n-1$.
		\end{center}
	\end{itemize}
	\begin{proof}
		We just prove Case (2), as the conclusion for Case (1) can be similarly proved.
		
		By the definition, the defining set of ${\C}_{(q,n,\delta,2)}$ with respect to $\beta_{1}$ is $T=C_{2}\cup C_{3}\cup\cdots\cup C_{\delta}$, $2\leq \delta\leq n-1$. Note that $0\notin T$, i.e., $0\in T^{\perp}$, this means $C_{0}\subset T^{\perp}$. Therefore, if $\C_{(q,n,\delta,2)}$ is a dually-BCH code, then there exists $r\geq 1$ such that $T^{\perp}=C_{0}\cup C_{1}\cup \cdots\cup C_{r-1}$.
		
		If $q\leq \delta\leq n-1$, then $C_{1}=C_{q}\subset T$, i.e., $T=C_{1}\cup C_{2}\cup\cdots\cup C_{\delta}$. By Lemma \ref{lem8}, we know that ${\C}_{(q,\frac{q^{2s}-1}{q+1},\delta,2)}$ is a dually-BCH code if and only if $\delta_{1,n}\leq \delta\leq n-1$.
		
		If $2\leq \delta\leq q-1$, note that $\frac{nq^{\lceil \frac{m}{2}\rceil}}{q^{m}-1}=\frac{q^{s}}{q+1}$. Since $s\geq 2$, then $\left[1,q-1\right]\subset\left[ 1,\frac{nq^{\lceil \frac{m}{2}\rceil}}{q^{m}-1}\right]$. Hence, $i \in\MinRep_n$ and $\mid C_{i}\mid =2s$ for all $1\leq i\leq q-1$ by Lemma \ref{lem1}. We have $C_{2}\subset T$ and $C_{1}\not\subset T$, i.e., $2s\leq \dim(\C_{(q,n,\delta,2)})\leq n-2s<n$. We consider two cases.
		\begin{itemize}
			\item [1)]If $2\nmid s$, note that $\frac{q^{s}+1}{q+1}$ is a coset leader modulo $n$ and $\frac{q^{s}+1}{q+1}>q-1\geq \delta$, i.e., $\frac{q^{s}+1}{q+1}\notin T$. Since $CL(n-\frac{q^{s}+1}{q+1})=CL(\frac{q^{2s}-q^{s}-2}{q+1})=\frac{(q-1)q^{2s-1}-q^{s-1}-1}{q+1}=\delta_{1,n}$, then $C_{\delta_{1,n}}\subset T^{\perp}$. Therefore, if $\C_{(q,n,\delta,2)}$ is a dually-BCH code, then $T^{\perp}=C_{0}\cup C_{1}\cup \cdots \cup C_{\delta_{1,n}}$. However, $\dim(\C_{(q,n,\delta,2)})\leq n-2s$, which contradicts the equation of $\dim(\C_{(q,n,\delta,2)})+\dim(\C^{\perp}_{(q,n,\delta,2)})=n$.
			\item [2)]If $2\mid s$, note that $\frac{q^{s-1}+1}{q+1}$ is a coset leader modulo $n$ and $\frac{q^{s-1}+1}{q+1}>q-1\geq \delta$, i.e., $\frac{q^{s-1}+1}{q+1}\notin T$. Since $CL(n-\frac{q^{s-1}+1}{q+1})=\frac{(q-1)q^{2s-1}-q^{s}-1}{q+1}=\delta_{1,n}$, then $C_{\delta_{1,n}}\subset T^{\perp}$. Hence, if $\C_{(q,n,\delta,2)}$ is a dually-BCH code, then $T^{\perp}=C_{0}\cup C_{1}\cup \cdots \cup C_{\delta_{1,n}}$. However, $\dim(\C_{(q,n,\delta,2)})\leq n-2s$, which contradicts the equation of $\dim(\C_{(q,n,\delta,2)})+\dim(\C^{\perp}_{(q,n,\delta,2)})=n$.
		\end{itemize}
		 Thus we complete the proof.
	\end{proof}
\end{theorem}
\begin{theorem}\label{th29}
Let $q=2$ and $s\geq 2$ be even. Then ${\C}_{(q,n,\delta,2)}$ is a dually-CBH code if and only if
	\begin{center}
		$\delta_{1,n}\leq \delta\leq n-1$.
	\end{center}
\begin{proof}
	The proof is very similar to that of Theorem \ref{th28}, hence we omit it.
\end{proof}
\end{theorem}
\section{The case of $n=\frac{q^{m}-1}{q-1}$}\label{set5}
In this section, we always suppose $n=\frac{q^m-1}{q-1}$ and $\beta_{2}=\alpha^{q-1}$, where $q\geq3$ and $m\geq 4$.
\subsection{The computation of $\delta_{2,n}$}
\begin{lemma}\label{lem30}
	Let $0\leq t<n$ and $t$ be of the form
	\begin{center}
		$t=(0,\underbrace{q-1,\ldots,q-1}_{n_{q-1}},\underbrace{q-2,\ldots,q-2}_{n_{q-2}},\ldots,\underbrace{i,\ldots,i}_{n_{i}},\ldots,\underbrace{1,\ldots,1}_{n_{1}}).$
	\end{center}
\begin{itemize}
	\item If $\sum_{i=l}^{q-1}n_{i}<j<\sum_{i=l+1}^{q-1}n_{i}$, where $l\in \lbrace1,2,\ldots,q-1\rbrace$, then $[tq^{j}]_{n}>t$.
	\item If $j=\sum_{i=l}^{q-1}n_{i}$, $q-1\geq l\geq 1$, then $$[tq^{j}]_{n}=(0,\underbrace{q-1,\ldots,q-1}_{n_{l-1}},\ldots,\underbrace{q-l+1,\ldots,q-l+1}_{n_{1}+1},\underbrace{q-l,\ldots,q-l}_{n_{q-1}},\ldots,\underbrace{1,\ldots,1}_{n_{l}-1}).$$
\end{itemize}
\begin{proof}
	Note that $\frac{q^m-1}{q-1}|{tq^{i}}-[tq^{i}]_{q^{m}-1}$, i.e., $[tq^{i}]_{q^{m}-1}\equiv[tq^{j}]_{n}$.
	\begin{itemize}
	\item[1)] If $\sum_{i=l}^{q-1}n_{i}>j>\sum_{i=l+1}^{q-1}n_{i}$, where $l\in \lbrace2,3,\ldots,q-1\rbrace$,
	then $[tq^{j}]_{n}$ is congruent to
	\begin{center}
		$[tq^{j}]_{q^{m}-1}=(\underbrace{l,\ldots,l}_{n_{l}-u},\underbrace{l-1,\ldots,l-1}_{n_{l-1}},\ldots,0,\underbrace{q-1,\ldots,q-1}_{n_{q-1}},\ldots,\underbrace{l,\ldots,l}_{u})$,
	\end{center}
	where $u=j-\sum_{i=l+1}^{q-1}n_{i}-1 ,0 \leq u <n_{l}-1$. Note that $ln>[tq^{j}]_{q^{m}-1}>n$, then $[tq^{j}]_{n}=[tq^{j}]_{q^{m}-1}-(l-1)n>q^{m-1}>t$.
	\\If $\sum_{i=1}^{q-1}n_{i}>j>\sum_{i=2}^{q-1}n_{i}$, then $[tq^{j}]_{n}$ is congruent to
	\begin{center}
		$[tq^{j}]_{q^{m}-1}=(\underbrace{1,\ldots,1}_{n_{l}-u},0,\underbrace{q-1,\ldots,q-1}_{n_{q-1}},\ldots,\underbrace{2,\ldots,2}_{n_{2}},\underbrace{l,\ldots,l}_{u})$,
	\end{center}
where $u=j-\sum_{i=2}^{q-1}n_{i}-1 ,0 \leq u <n_{1}-1$. Note that $[tq^{j}]_{q^{m}-1}<n$, then $[tq^{j}]_{n}=[tq^{j}]_{q^{m}-1}>q^{m-1}>t$.
	\item[2)] If $j=\sum_{i=l}^{q-1}n_{i}$, $q-1\geq l\geq 1$, then
	\begin{align}
		[tq^{j}]_{n}=&[tq^{j}]_{q^{m}-1}-(l-1)n\notag \\
		=&(0,\underbrace{q-1,\ldots,q-1}_{n_{l-1}},\ldots,\underbrace{q-l+1,\ldots,q-l+1}_{n_{1}+1},\underbrace{q-l,\ldots,q-l}_{n_{q-1}},\ldots,\underbrace{1,\ldots,1}_{n_{l}-1}). \label{con:e10}
	\end{align}
\end{itemize} Thus we complete the proof.
\end{proof}
\end{lemma}

\begin{lemma}\label{lem31}
Let $q>3$ and $m\geq q$. Suppose $m-1=a(q-1)+b$, where $a\geq 1$ and $0\leq b\leq q-2$.
\begin{itemize}
\item[(1)] If $a\geq 3$ and $b=0$, i.e., $m=a(q-1)+1$, then
\begin{center}
		$\delta_{2,n}=\frac{q^m-1-q^{m-1}-q^{m-a}-\Sigma_{l=1}^{q-3}q^{al-1} }{q-1},$
\end{center}
and $\mid C_{\delta_{2,n}}\mid =m$.
\item[(2)] If $b=1$, i.e., $m=a(q-1)+2$, let $A=\left\lfloor\frac{q-1}{2}\right\rfloor$, then
\begin{center}
	$\delta_{2,n}=\frac{q^m-1-q^{m-1}-\Sigma_{l=1}^{A-1}q^{al}-\Sigma_{l=A}^{q-2}q^{al+1} }{q-1},$
\end{center}
and  $\mid C_{\delta_{2,n}}\mid =m$.
\item[(3)] If $b=2$, i.e., $m=a(q-1)+3$, let $A=\left\lfloor\frac{q-1}{3}\right\rfloor$, then
\begin{center}
	 $\delta_{2,n}$=$\begin{cases}
\frac{q^{m}-1-\sum_{t=1}^{q-1}q^{\lceil \frac{mt}{q-1}-1\rceil}}{q-1}-q^{(2A+1)a+1}+q^{(A+1)a},& \text{if $q\equiv0\,(mod\,3)$} ;\\
\frac{q^{m}-1-\sum_{t=1}^{q-1}q^{\lceil \frac{mt}{q-1}-1\rceil}}{q-1}-q^{2Aa+1},& \text{if $q\equiv1\,(mod\,3)$};\\
\frac{q^{m}-1-\sum_{t=1}^{q-1}q^{\lceil \frac{mt}{q-1}-1\rceil}}{q-1}-q^{Aa},& \text{if $q\equiv2\,(mod\,3)$},
	\end{cases}$
\end{center}
and  $\mid C_{\delta_{2,n}}\mid =m$.
\end{itemize}
\begin{proof}
We just provide the proof for Case (1), since the proofs for the other cases are similar.

Note that
\begin{center}
$\frac{q^m-1-q^{m-1}-q^{m-a}-\Sigma_{l=1}^{q-3}q^{al-1} }{q-1}=(0,\underbrace{q-1,\ldots,q-1}_{a-1},\underbrace{q-2,\ldots,q-2}_{a+2},\ldots,\underbrace{i,\ldots,i}_{a},\ldots,\underbrace{1,\ldots,1}_{a-1})$.
\end{center}
It then follows that $\delta=\frac{q^m-1-q^{m-1}-q^{m-a}-\Sigma_{l=1}^{q-3}q^{al-1} }{q-1}\in \MinRep_n$ and $\mid C_{\delta}\mid =m$ by Lemma \ref{lem30}.

We claim that $\delta_{2,n}=\frac{q^m-1-q^{m-1}-q^{m-a}-\Sigma_{l=1}^{q-3}q^{al-1} }{q-1}$. Suppose there exists an integer $s$ such that $s\in\MinRep_n$ and $\delta_{2,n} <s<\delta_{1,n} $. Let the $q$-adic expansion of $s$ be $\sum_{i=1}^{m-1}s_{i}q^{i}$. Note that $$\delta_{1,n}=(0,\underbrace{q-1,\ldots,q-1}_{a},\underbrace{q-2,\ldots,q-2}_{a},\ldots,\underbrace{i,\ldots,i}_{a},\ldots,\underbrace{1,\ldots,1}_{a}).$$ 
Then we have $s_{m-1}=0$, $s_{i}=q-1$ for all $m-a\leq i\leq m-2$ and $s_{m-1-a}=q-1$ or $s_{m-1-a}=q-2$.
\begin{itemize}
\item[1)] If $s_{m-1-a}=q-1$, then $i_{l}=q-1$ for all $m-1-a\leq l\leq m-2$. By Lemma \ref{lem4}, $s$ must be of the form $$(0,\underbrace{q-1,\ldots,q-1}_{a},\underbrace{q-2,\ldots,q-2}_{n_{q-2}},\ldots,\underbrace{i,\ldots,i}_{n_{i}},\ldots,\underbrace{1,\ldots,1}_{n_{1}}).$$
Since $[sq^{j}]_{n}\geq s$ for $j=\sum_{i=l}^{q-1}n_{i}$, then $n_{l}\geq a$ for all $2\leq l\leq q-2$ and $n_{1}\geq a-1$ by Eq. (\ref{con:e10}). In addition, $m=a(q-1)+1$, then $n_{1}=a$ or $n_{1}=a-1$. Hence, if $n_{1}=a$ then $s=\delta_{1,n}$; if $n_{1}=a-1$ then $s>\delta_{1,n}$, which contradicts the fact that $s<\delta_{1,n}$.
\item[2)] If $s_{m-1-a}=q-2$, we divide into two steps to prove.

$\qquad$ Step 1. We claim that $s$ is of the form
\begin{equation}\label{eq:14} (0,\underbrace{q-1,\ldots,q-1}_{a-1},\underbrace{q-2,\ldots,q-2}_{n_{q-2}},\ldots,\underbrace{i,\ldots,i}_{n_{i}},\ldots,\underbrace{1,\ldots,1}_{n_{1}}).
\end{equation}
If there exists an integer $i$ such that $s_{i}\neq 0$ and $s_{i}< s_{i-1}$, then we have $[sq^{m-1-i}]_{n}=[sq^{m-1-i}]_{q^{m}-1}-s_{i}n<s$, which contradicts the fact that $s\in \MinRep_n$. Hence we konw that $s_{i}\geq s_{i-1}$ if $s_{i}\neq 0$. Then $s=(I_{e},I_{e-1},\ldots,I_{0})$, where  $$I_{t}=(0,\underbrace{q-1,\ldots,q-1}_{n_{t,q-1}},\underbrace{q-2,\ldots,q-2}_{n_{t,q-2}},\ldots,\underbrace{1,\ldots,1}_{n_{t,1}})$$
for all $0\leq t\leq e$ and $n_{t,j}\geq 0$ for all $1\leq j\leq q-1$. Clearly, $n_{e,q-1}=a-1$.

$\qquad$ Let $k=\sum_{j=1}^{q-1}n_{0,j}$. Since $[sq^{m-1-k}]_{n}\geq s$, then $n_{0,q-1}\geq n_{e,q-1}=a-1$. Similarly, we get $n_{t,q-1}\geq a-1$ for all $e\geq t\geq 0$.
Denote the $q$-adic expansion of $[sq^{n_{e,q-1}}]_{n}$ by $\sum_{i=1}^{m-1}s^{'}_{i}q^{i}$, then $(s^{'}_{m-1},s^{'}_{m-2},\ldots,s^{'}_{1})$ must be of the form $(I^{'}_{e},I^{'}_{e-1},\ldots,I^{'}_{0})$ by Eq. (\ref{con:e10}), where $$I^{'}_{e}=(0,\underbrace{q-1,\ldots,q-1}_{n_{e,q-2}},\ldots,\underbrace{3,\ldots,3}_{n_{e,2}},\underbrace{2,\ldots,2}_{n_{e,1}+1},\underbrace{1,\ldots,1}_{n_{e-1,q-1}-1}),$$
$$I^{'}_{0}=(0,\underbrace{q-1,\ldots,q-1}_{n_{0,q-2}},\ldots,\underbrace{3,\ldots,3}_{n_{0,2}},\underbrace{2,\ldots,2}_{n_{0,1}+1},\underbrace{1,\ldots,1}_{n_{e,q-1}-1}).$$
Since $[sq^{n_{e,q-1}}]_{n}\geq s$, then $n_{e,q-2}\geq a-1$.

$\qquad$ Similarly, we have $n_{t,j}\geq a-1$ for all $e\geq t\geq 0$ and $2\leq j\leq q-2$ and $n_{t,1}\geq a-2$ for all $e\geq t\geq 0$. Therefore,
$$m=\sum_{t=0}^{e}\left(\sum_{j=1}^{q-1}n_{t,j}+1 \right)\geq (e+1)(q-1)(a-1).$$
 If $e\geq 1$, we have $m=a(q-1)+1\geq 2(q-1)(a-1)$, which contradicts to $a\geq 3$.
Hence, $s$ is of the form Eq. (\ref{eq:14}), 
where $n_{i}\geq a-1$ for all $q-2\geq i\geq 2$ and $n_{1}\geq a-2$.

$\qquad$ Step 2. Since $s> \delta_{2,n}$, then $n_{q-2}\geq a+3$, or $n_{q-2}=a+2$ and there exists an integer $i$ such that $n_{i}\geq a+1$ and $n_{j}=a$ for all $q-3\geq j>i$.

$\qquad$(2.1) If $n_{q-2}\geq a+3$, since $m=a(q-1)+1$, then there exists $n_{j}= a-1$ with $q-3\geq j \geq 2$, or $n_{t}=a$ for all $q-3\geq t\geq 2$ and $n_{1}=a-2$.

$\qquad$If there exists $n_{j}= a-1$ with $q-3\geq j \geq 2$. Let $\mid \lbrace n_{i}\geq a+1 : 1\leq i < q-2\rbrace\mid =k$, then $\mid \lbrace n_{i}=a-1 : 1\leq i < q-2\rbrace\mid \geq k+2$. Hence, there are two integers $u$ and $v$ such that $n_{u},n_{v}=a-1$ and $n_{i}\leq a$ for all $u\geq i\geq v$. Note that
$$\left[sq^{\sum_{i=u+1}^{q-1}n_{i}}\right ]_{n}=(0,\underbrace{q-1,\ldots,q-1}_{a-1},\underbrace{q-2,\ldots,q-2}_{n_{u-1}},\ldots,\underbrace{q+v-u-1,\ldots,q+v-u-1}_{a-1},\ldots,\underbrace{1,\ldots,1}_{n_{u+1}-1}).$$
Clearly, $\left[sq^{\sum_{i=u+1}^{q-1}n_{i}}\right ]_{n}<s$, which contradicts the fact that $s\in \MinRep_n$.

$\qquad$If $n_{t}=a$ for all $q-3\geq t\geq 2$ and $n_{1}=a-2$, we have $\left[sq^{\sum_{i=2}^{q-1}n_{i}}\right ]_{n}<s$, which contradicts the fact that $s\in \MinRep_n$.

$\qquad$(2.2) If $n_{q-2}=a+2$ and there exists an integer $i$ such that $n_{i}\geq a+1$, $n_{j}=a$ for all $q-3\geq j>i$, the case is similar to (2.1).
\end{itemize}
Thus we complete the proof.
\end{proof}
\end{lemma}
\begin{lemma}\label{lem32}
Let $q>3$ and $m\geq q$. Suppose $m-1=a(q-1)+b$, where $a\geq 1$ and $0\leq b\leq q-2$.
\begin{itemize}
	\item [(1)]If $b=q-4$, i.e., $m=a(q-1)+q-3$, let $A=\left\lfloor\frac{q}{2}\right\rfloor$, then
	\begin{center}
		$\delta_{2,n}$=
		$\frac{q^{m}-1-\sum_{t=1}^{q-1}q^{\lceil \frac{mt}{q-1}-1\rceil}}{q-1}-q^{A(a+1)-2},$	
	\end{center}
	and $\mid C_{\delta_{2,n}}\mid=m$.
	\item [(2)]If $b=q-3$, i.e., $m=a(q-1)+q-2$, then
	\begin{center}
		$\delta_{2,n}=\frac{q^{m}-1-\sum_{t=1}^{q-1}q^{\lceil \frac{mt}{q-1}-1\rceil}}{q-1}-q^{a},$
	\end{center}
	and $\mid C_{\delta_{2,n}}\mid=m$.
	\item [(3)]If $b=q-2$, i.e., $m=(a+1)(q-1)$, then
	\begin{center}
		$\delta_{2,n}=\frac{q^m-1-q^{m-1}-q^{m-1-a}-\Sigma_{l=1}^{q-3}q^{(a+1)l-1} }{q-1},$
	\end{center}
	and $\mid C_{\delta_{2,n}}\mid=m$.
\end{itemize}
\begin{proof}
We just give the proof for Case (1), since the proofs for the other cases are similar.

If $2\nmid q$, then $\delta=\frac{q^{m}-1-\sum_{t=1}^{q-1}q^{\lceil \frac{mt}{q-1}-1\rceil}}{q-1}-q^{A(a+1)-2}$ is of the form
\begin{center}
$(0,\underbrace{q-1,...,q-1}_{a},\ldots,\underbrace{t,\ldots,t}_{a+1},\ldots,\underbrace{\frac{q+1}{2},\ldots,\frac{q+1}{2}}_{a},\underbrace{\frac{q-1}{2},\ldots,\frac{q-1}{2}}_{a+1},\ldots,\underbrace{i,\ldots,i}_{a+1},\ldots,\underbrace{1,\ldots,1}_{a}).$
\end{center}
It then follows that $\delta\in \MinRep_n$ and $\mid C_{\delta}\mid=m$ by Lemma \ref{lem30}.

We claim that $\delta_{2,n}$=
$\frac{q^{m}-1-\sum_{t=1}^{q-1}q^{\lceil \frac{mt}{q-1}-1\rceil}}{q-1}-q^{A(a+1)-2}$. 
Suppose there exists an integer $s$ such that $s\in \MinRep_n$ and $\delta_{2,n} <s<\delta_{1,n} $. Let the $q$-adic expansion of $s$ be $\sum_{i=1}^{m-1}s_{i}q^{i}$. Note that  $$\delta_{1,n}=(0,\underbrace{q-1,\ldots,q-1}_{a},\ldots,\underbrace{t,\ldots,t}_{a+1},\ldots,\underbrace{\frac{q+1}{2},\ldots,\frac{q+1}{2}}_{a+1},\underbrace{\frac{q-1}{2},\ldots,\frac{q-1}{2}}_{a},\ldots,\underbrace{i,\ldots,i}_{a+1},\ldots,\underbrace{1,\ldots,1}_{a}).$$
Since $\delta_{2,n} <s<\delta_{1,n}$, we have $s_{i}=q-1$ for all $m-1-a\leq i\leq m-2$. By Lemma \ref{lem4}, $s$ must be of the form
$$(0,\underbrace{q-1,\ldots,q-1}_{a},\ldots,\underbrace{t,\ldots,t}_{a+1},\ldots,\underbrace{\frac{q+1}{2},\ldots,\frac{q+1}{2}}_{n_\frac{q+1}{2}},\underbrace{\frac{q-1}{2},\ldots,\frac{q-1}{2}}_{n_{\frac{q-1}{2}}},\ldots,\underbrace{i,\ldots,i}_{n_{i}},\ldots,\underbrace{1,\ldots,1}_{n_{1}}).$$
Note that
$$\left[sq^{\sum_{i=2}^{q-1}n_{i}}\right]_{n}=(0,\underbrace{q-1,\ldots,q-1}_{n_{1}+1},\underbrace{q-2,\ldots,q-2}_{a},\ldots,\underbrace{t,\ldots,t}_{a+1},\underbrace{\frac{q-1}{2},\ldots,\frac{q-1}{2}}_{n_{\frac{q+1}{2}}},\ldots,\underbrace{1,...,1}_{n_{2}-1}),$$
then $n_{1}\geq a$ since $\left[sq^{\sum_{i=2}^{q-1}n_{i}}\right]_{n}\geq s$. Similarly, we have $n_{i}\geq a$ for all $\frac{q+1}{2}\geq i\geq 1$. Since $\delta_{2,n} <s<\delta_{1,n}$, then $n_{\frac{q+1}{2}}=a+1$ or $n_{\frac{q+1}{2}}=a$.
\begin{itemize}
	\item[1)] When $n_{\frac{q+1}{2}}=a+1$. Since $\delta_{1,n}>s$, then there exists an integer $l$ such that $n_{l}= a$ and $n_{j}=a+1$ for all $\frac{q-1}{2}> j>l>1$.
	Note that
	$$\left[sq^{\sum_{i=\frac{q+1}{2}}^{q-1}n_{i}}\right]_{n}=(0,\underbrace{q-1,\ldots,q-1}_{a},\ldots,\underbrace{t,\ldots,t}_{a+1},\ldots,\underbrace{l+\frac{q-1}{2},\ldots,l+\frac{q-1}{2}}_{a},\ldots,\underbrace{1,\ldots,1}_{a}),$$
	Since $q-1-\frac{q-1}{2}>q-1-\frac{q-1}{2}-l$, we have $\left[sq^{\sum_{i=\frac{q+1}{2}}^{q-1}n_{i}}\right]_{n}< s$, which contradicts the fact that $s$ is a coset leader.
	\item[2)] When $n_{\frac{q+1}{2}}=a$.
	
	If $n_{\frac{q-1}{2}}\geq a+2$, there are two integers $u$ and $v$ such that $n_{u},n_{v}=a$ and $n_{i}\leq a+1$ for all $\frac{q-3}{2}\geq u\geq i\geq v\geq 1$. Note that 
	$$\left[sq^{\sum_{i=u+1}^{q-1}n_{i}}\right ]_{n}=
	(0,\underbrace{q-1,\ldots,q-1}_{a},\underbrace{q-2,\ldots,q-2}_{n_{u-1}},\ldots,\underbrace{q+v-u-1,\ldots,q+v-u-1}_{a},\ldots,\underbrace{1,\ldots,1}_{n_{u+1}-1}).$$
	Since $q-1-\frac{q+1}{2}>u-v$, then we have $\left[sq^{\sum_{i=u+1}^{q-1}n_{i}}\right ]_{n}<s$, which contradicts the fact that $s$ is a coset leader.
	
	If $n_{\frac{q-1}{2}}= a+1$, there must exist an integer $l$ such that $n_{l}\geq a+2$ and $n_{i}= a+1$ for all $\frac{q-3}{2}\geq i>l>1$. We have $\mid\lbrace n_{i}=a : l> i \geq 1\rbrace\mid \geq 2$. Hence, there are two integers $u$ and $v$ such that $n_{u},n_{v}=a$ and $n_{i}\leq a+1$ for all $l>u\geq i\geq v\geq 1$. We can get that $\left[sq^{\sum_{i=u+1}^{q-1}n_{i}}\right ]_{n}<s$, which contradicts the fact that $s$ is coset leader.
\end{itemize}

In the same way, we can prove $\delta_{2,n}$ is the second largest coset leader when $2\mid q$. This completes the proof.
\end{proof}
\end{lemma}
\subsection{BCH Codes and Dually-BCH Codes}
Then we have the following conclusion when the length of the BCH codes $\C_{(q,n,\delta)}$ satisfy the cases of Lemmas \ref{lem31} and \ref{lem32}.

\begin{theorem}\label{th33}
Let $q>3$ and $m$ be of the form given by Lemmas \ref{lem31} or \ref{lem32}. Then the BCH code $\C_{(q,n,\delta)}$ with $\delta_{2,n}+1\leq \delta \leq \delta_{1,n}$ has parameters
\begin{center}
	$[\frac{q^{m}-1}{q-1},\frac{m}{gcd(m,q-1)}+1,d\geq \delta]$,
\end{center}
and the BCH code $\C_{(q,n,\delta_{\delta_{2,n}})}$ has parameters
\begin{center}
	$[\frac{q^{m}-1}{q-1},\frac{m}{gcd(m,q-1)}+m+1,d\geq \delta_{2,n}]$,
\end{center}
\begin{proof}
	Form  Lemmas \ref{lem4}, \ref{lem31} and \ref{lem32}, we can obtain the results directly.
\end{proof}
\end{theorem}
\begin{example}
	 Let $(q,m)=(4,5)$, we have $n=341$. The BCH code $\C_{(4,341,\delta)}$ with $230\leq \delta\leq 233$ has parameters $[341,6,\geq \delta]$, and the code $\C_{(4,341,229)} $ has parameters $[341,11,\geq 229]$.
\end{example}
\begin{lemma}\label{lem35}
	Let $q\geq 3$ and $m\geq 4$. If $2\leq \delta\leq q-1$, then we have $\delta_{1,n}\in T^{\perp}$.
	\begin{proof}
		Let $m=a(q-1)+b$, $a\geq 0$ and $0\leq b\leq q-2$. Since $T^{\perp}=\Z_{n}$$\backslash T^{-1}$, then $C_{\delta_{1,n}}\subset T^{\perp}$ if and only if $C_{(n-\delta_{1,n})}\nsubseteq T$. By Lemma \ref{lem4}, we have $(q-1)(n-\delta_{1,n})=\sum_{t=1}^{q-1}q^{\lceil \frac{mt}{q-1}-1\rceil}=\sum_{i=0}^{m-1}a_{i}q^{i}$, where $a_{i}=\left\lceil \frac{q-1}{m}\right\rceil$ or $a_{i}=\left\lfloor \frac{q-1}{m}\right\rfloor$. Let $\delta^{'}\in \MinRep_{q^{m}-1}$ and $a\in C_{(q-1)(n-\delta_{1,n})}$.
		
		If $a=0$, then $q-1>m$. Note that $a_{i}=\left\lceil \frac{q-1}{m}\right\rceil\geq 2$ or $a_{i}=\left\lfloor \frac{q-1}{m}\right\rfloor\geq 1$ for all $0\leq i\leq m-1$, then $\delta^{'}\geq\sum_{i=0}^{m-1}q^{i}$ by Lemma \ref{lem4}. Then
		$$CL(n-\delta_{1,n})=\frac{\delta^{'}}{q-1}\geq \frac{\sum_{i=0}^{m-1}q^{i}}{q-1}>q+1>\delta.$$
		
		If $a\geq 1$, similarly, we have $CL(n-\delta_{1,n})$ has the form $(\ldots,1,\underbrace{0,\ldots,0}_{a})$, i.e., $CL(n-\delta_{1,n})>q^{m-a-1}-1$. It is easy to check that $m-a-1\geq a+1$ when $m\geq4$. We have
			$$CL(n-\delta_{1,n})=\frac{\delta^{'}}{q-1}>\frac{q^{m-a-1}}{q-1}\geq \frac{q^{a+1}-1}{q-1}>q+1>\delta.$$
		Therefore, $C_{(n-\delta_{1,n})}\nsubseteq T$, i.e., $\delta_{1,n}\in T^{\perp}$. This completes the proof.
	\end{proof}
\end{lemma}
For $b=1$, the conditions of $\C_{(q,n,\delta)}$ being dually-BCH codes have been given by Lemma~\ref{lem8}. For $b=2$, we have
\begin{theorem}\label{th36}
	Let $q\geq 3$ and $m\geq 4$, then ${\C}_{(q,n,\delta,2)}$ is a dually-CBH code if and only if
	\begin{center}
		$\frac{q^{m}-1-\sum_{t=1}^{q-1}q^{\left\lceil \frac{mt}{q-1}-1\right\rceil}}{q-1}\leq \delta \leq n-1$.
	\end{center}
	\begin{proof}
		The proof is very similar to that of Theorem \ref{th28}, we omit it.	
	\end{proof}
\end{theorem}

\section{Conclusions}\label{set6} The main contributions of this paper are as follows:
\begin{itemize}
	\item For the codes of length $n=q^{m}-1$, we found the $i$-th largest $q$-cyclotomic coset leader is $\delta_{i}=(q-1)q^{m-1}-1-q^{\lfloor \frac{m-1}{2} \rfloor +i-2}$. The parameters of $\C_{(q,n,\delta_{i,n})}$ was investigated (see Theorem \ref{th11}).
	
	\item For the codes of length $n=\frac{q^{2s}-1}{q+1}$, we find the second largest coset leader $\delta_{2,n}$. The parameters of $\C_{(q,n,\delta)}$ with $ \delta_{2,n}\leq \delta\leq \delta_{1,n}$ and its dual code was settled (see Theorems \ref{th18}-\ref{th20}).
	The dimension of $\C_{(q,n,\delta)}$ were determined, where $\lceil \frac{q}{2}\rceil q^{s-1}\leq \delta \leq \frac{q^{s+1}+1}{q+1}$ and $2\mid s$(see Theorem \ref{th24}). Finally, we gave the dimension and the minimum distance of three subclasses of $\C_{(q,n,\delta)}$ for $\delta=a\frac{q^{s}-1}{q-1},a\frac{q^{s}+1}{q+1}$ if $2\nmid s$ and $\delta=a\frac{q^{s}-1}{q^{2}-1}$ if $2\mid s$, $1\leq a\leq q-1$ (see Theorem \ref{th26}).
	
	\item For the codes of length $n=\frac{q^{m}-1}{q-1}$, we found the second largest coset leader $\delta_{2,n}$ for some special cases. The parameters of $\C_{(q,n,\delta)}$ with $ \delta_{2,n}\leq \delta\leq \delta_{1,n}$ were investigated (see Theorem \ref{th33}).
	
	\item Sufficient and necessary conditions for $\C_{(q,n,\delta,2)}$ being dually-BCH codes were given, where $n=q^{m}-1,\frac{q^{m}-1}{q-1}$ and $\frac{q^{2s}-1}{q+1}$ (see Theorems \ref{th12}, \ref{th13}, \ref{th28}, \ref{th29}, \ref{th36}). Moreover, we found the sufficient and necessary conditions for the dual code of $\widetilde{C}_{(q,q^{m}-1,\delta)}$ to be a narrow-sence primitive BCH code (see Theorems \ref{th14} and \ref{th15}).
\end{itemize}


\end{document}